\documentclass[11pt,a4paper]{amsart}

\usepackage[english]{babel}

\usepackage[top=2cm,bottom=2cm,left=3cm,right=3cm,marginparwidth=1.75cm]{geometry}

\usepackage{subcaption}
\usepackage{tabularx}
\usepackage{amsmath,amssymb,bm,color}
\usepackage[english]{babel}
\usepackage{mathtools}
\usepackage{enumerate}
\usepackage[numbers]{natbib}

\usepackage{subcaption}
 
\usepackage[title]{appendix} 
\makeatletter
\newcommand{\@chapapp}{\relax}%
\makeatother
\usepackage{multirow}
\usepackage{xcolor}
\usepackage{algorithm}
\usepackage[noend]{algpseudocode}
\usepackage{float}
\usepackage[section]{placeins}

\usepackage{pgfplots,relsize}

\newcommand{\lweight}[1]{\mathrm{wt}_{L} \left( #1 \right) }

 \newtheorem{theorem}{Theorem}
\newtheorem{problem}{Problem}
 
 \newtheorem{lemma}[theorem]{Lemma}
 
 \newtheorem{corollary}[theorem]{Corollary}
 \theoremstyle{definition}
 
 \newtheorem{definition}[theorem]{Definition}
 \theoremstyle{remark}
 \newtheorem{remark}[theorem]{Remark}

\numberwithin{equation}{section}

\newcommand{\zps}{\mathbb{Z} /p^s \mathbb{Z}}
\newcommand{\zpsk}[2][s]{\left(\mathbb{Z}/p^{#1}\mathbb{Z}\right)^{#2}}

\newcommand{\bA}{\mathbf{A}}
\newcommand{\bB}{\mathbf{B}}
\newcommand{\bH}{\mathbf{H}}
\newcommand{\bG}{\mathbf{G}}
\newcommand{\bs}{\mathbf{s}}

\newcommand{\bc}{\mathbf{c}}

\newcommand{\be}{\mathbf{e}}

\newcommand{\bt}{\mathbf{t}}
\newcommand{\by}{\mathbf{y}}
\newcommand{\bx}{\mathbf{x}}

\newcommand{\bU}{\mathbf{U}}
\newcommand{\bP}{\mathbf{P}}

\newcommand{\concat}{%
  \mathbin{{+}\mspace{-8mu}{+}}%
}

\usepackage{stmaryrd}

\newcommand{\set}[1]{\left\lbrace #1 \right\rbrace}
\newcommand{\card}[1]{\left\vert #1 \right\vert}

\newcommand{\prob}{\mathbb{P}} 

\usepackage{todonotes}

\usepackage{tikz}
\usetikzlibrary{shapes,arrows}
\usetikzlibrary{calc}
\usetikzlibrary{patterns} 
\usepackage{mathrsfs}
\usetikzlibrary{fadings}
\usetikzlibrary{shadows.blur}

\tikzset{caption/.style={insert path={
let \p1=($(current bounding box.east)-(current bounding box.west)$) in
(current bounding box.south) node[below,text width=\x1-4pt,align=center] 
{\captionof{figure}{#1}}}}}
\DeclareMathOperator{\supp}{supp}

\title{Information Set Decoding for Lee-Metric Codes using Restricted Balls}
 
\begin{document}
	\author[J. Bariffi]{Jessica Bariffi}
	\address{Institute of Communication and Navigation \\ German Aerospace Center \\
 Germany
	}
	\email{jessica.bariffi@dlr.de}
	
	\author[K. Khathuria]{Karan Khathuria}
	\address{Institute of Computer Science\\ University of Tartu \\ Estonia
	}
	\email{karan.khathuria@ut.ee}
	
	\author[V. Weger]{Violetta Weger}
	\address{Department of Electrical and Computer Engineering\\
		Technical University of Munich\\
 Germany
	}
	\email{violetta.weger@tum.de}

\maketitle

\begin{abstract}
The Lee metric syndrome decoding problem is an NP-hard problem and several generic decoders have been proposed. The observation that such decoders come with a larger cost than their Hamming metric counterparts make the Lee metric a promising alternative for classical code-based cryptography.  Unlike in the Hamming metric, an error vector that is chosen uniform at random of a given Lee weight is expected to have only few entries with large Lee weight. Using this expected distribution of entries, we are able to drastically decrease the cost of generic decoders in the Lee metric, by reducing the original problem to a smaller instance, whose solution lives in restricted balls. 

\end{abstract}

\section{Introduction}
The original syndrome decoding problem (SDP) asks to decode a random linear code over a finite field endowed with the Hamming metric. This problem has been long studied and is well understood. The SDP is an NP-hard problem \cite{Barg, NP} and lays the foundation of code-based cryptography, which is a promising candidate for post-quantum cryptography. 
The fastest algorithms to solve the syndrome decoding problem are called information set decoding (ISD) algorithms and  started with the work of Prange \cite{prange} in 1962.  Although the  literature on ISD algorithms in this classical case is vast  (see \cite{BJMM,  bernstein2011smaller, canteaut1998new, canteautsendrier, chabaud, finiasz,  Lee1988, Leon1988, mmt, stern}), the cost of generic decoding has only decreased little and is considered stable.  The fastest algorithm over the binary until this day is called BJMM algorithm \cite{BJMM} and uses the idea of representation technique from \cite{rep}. 
For an overview of the binary case see \cite{meurer2013coding}. With new cryptographic schemes proposed over general finite fields, most of these algorithms have been generalized to $\mathbb F_q$ (see \cite{  klamti, hirose, interlando2018generalization,  niebuhr, peters}). 

Due to new challenges in code-based cryptography, such as the search for efficient signature schemes, also other metrics are now investigated. 
For example the rank metric has gained a lot of attention due to the NIST submission ROLLO \cite{rollo} and RQC \cite{RQC}. 
While the understanding on the hardness of the rank-metric SDP is still rapidly developing  (see the new benchmark achieved in \cite{minrank}), it is still unknown whether the rank-metric SDP is an NP-hard problem.

The situation for the Lee metric is quite different.  The Lee-metric SDP was first studied for codes over $\mathbb{Z}/4\mathbb{Z}$ in \cite{Z4}. Later, in \cite{leenp} the problem was shown to be NP-hard over any $\mathbb{Z}/p^s\mathbb{Z}$ and several generic decoding algorithms to solve the problem have been provided. Also the paper \cite{thomas} confirmed the cost regimes of \cite{leenp} and more importantly the observation, that Lee-metric ISD algorithms cost more than their Hamming metric counterparts for fixed input parameters. Thus, the Lee metric has a great potential to reduce the key sizes or signature sizes in code-based cryptosystems. Modern code-based cryptography is moving away from the classical idea of McEliece \cite{mceliece}, where the distinguishability of the secret code obstructs a security reduction to the SDP, and moving towards ideas from lattice-based cryptography such as the ring learning with error (RLWE) problem. Note, that the Lee metric is the closest metric in coding theory to the Euclidean metric used in lattice-based cryptography, in the sense that both metrics take into consideration the magnitude of the entries. 

In this paper we use the new results from \cite{distribution} on the marginal distribution of vectors of given Lee weight to reduce the cost of the Lee-metric ISD algorithms further and thus contribute to the recent advances in understanding the hardness of this problem, with the final goal to deem this setting secure for applications. 

In fact, for the Lee-metric SDP we assume that the instance is given  by a randomly chosen parity-check matrix and an error vector of fixed Lee weight which was also chosen uniformly at random. The results from \cite{distribution} now provide us with new and central information on the sought-after error vector $\be$. In fact, using the marginal distribution we are able to determine the expected number of entries of $\be$, which have a fixed Lee weight. 
The main idea of the novel algorithm is that we expect only very few entries of $\be$ to have a large Lee weight (which is defined through a threshold $r$) if the relative Lee weight is lower than a fixed constant depending on the size $p^s$ of the residue ring $\zps$. Thus, using the partial Gaussian elimination (PGE) approaches from the classical case, like in BJMM \cite{BJMM}, we are able to reduce the original instance to a smaller instance, where the sought-after smaller error vector now only has entries of Lee weight up to $r$ and thus lives in smaller Lee-metric balls. This will clearly help reducing the cost of ISD algorithms. Similar ideas, that is to use the more reliable parts of an error vector to decode,  have been used since long (see e.g. \cite{fossorier}).
This paper thus reduces the cost of the algorithms from \cite{leenp} and \cite{thomas}, which were the fastest known Lee-metric ISD algorithms up to now.

This paper is organized as follows. In Section \ref{sec:prelim} we introduce the required notions on ring-linear codes and results for Lee-metric codes, such as the asymptotic of restricted spheres. In Section \ref{sec:distr}  we recall the results of \cite{distribution} on the marginal distribution and introduce the necessary values  for our algorithm.
The main part, the new Lee-metric ISD algorithm, is presented in Section \ref{sec:algo} together with an asymptotic cost analysis. The analysis considers the average time complexity in the classical and the quantum case. In addition, a reversed algorithm is presented in Section \ref{bigball}, where we decode  beyond the minimum distance. Finally, in Section \ref{sec:comparison} we compare the new algorithm to the previously fastest Lee-metric ISD algorithms.  

\section{Preliminaries}\label{sec:prelim}

\paragraph{Notation:} Let $p$ be a prime and $s$ be a positive integer and let us consider the integer residue ring $\mathbb{Z}/p^s\mathbb{Z}$. 
The cardinality of a set $V$ is denoted as $\card{V}$ and its complement by $V^C$.
We use bold lower case (respectively, upper case) letters to denote vectors  (respectively, matrices). By abuse of notation, a tuple in a module over a ring will still be denoted by a vector. 
The $n \times n$ identity matrix will be denoted by $\text{Id}_n.$ 
Let  $S \subseteq \{1, \ldots, n\}$.
For a vector $\bx \in \left(\mathbb{Z}/p^s\mathbb{Z}\right)^n$, we denote by $\bx_S$ the vector consisting of the entries of $\bx$ indexed by $S.$ Similarly, for a matrix $\bA \in \left(\mathbb{Z}/p^s\mathbb{Z}\right)^{k \times n}$, we denote by $\bA_S$ the matrix consisting of the columns of $\bA$ indexed by $S.$ Finally the symmetric group of $\{1, \ldots, n \}$ is denoted by $S_n.$

\begin{definition}
A \textbf{linear code} $\mathcal{C} \subseteq \left(\mathbb{Z}/p^s\mathbb{Z}\right)^n$ is a $\mathbb{Z}/p^s\mathbb{Z}$-submodule of $\left(\mathbb{Z}/p^s\mathbb{Z}\right)^n$.
\end{definition}
Since we are over a ring, our code does not possess a dimension, instead we denote by the \textbf{$\mathbb{Z}/p^s\mathbb{Z}$-dimension} of the code $\mathcal{C} \subseteq \left(\mathbb{Z}/p^s\mathbb{Z}\right)^n$ the following 
\[k \coloneqq \log_{p^s}\left(\mid \mathcal{C} \mid \right),\] such that the \textbf{rate} of the code is given by $R=\frac{k}{n}.$ In addition to the $\mathbb{Z}/p^s\mathbb{Z}$-dimension, the code $\mathcal{C} \subseteq \left(\mathbb{Z}/p^s\mathbb{Z}\right)^n$  also possesses a \textbf{rank} $K$, which is defined as the minimal number of generators of $\mathcal C$ as a $\zps$-module. In the case of a non-free code, note that $k<K.$

As for classical codes, we still have the notion of generator matrix and parity-check matrix.
\begin{definition}
Let  $\mathcal{C} \subseteq \left(\mathbb{Z}/p^s\mathbb{Z}\right)^n$  be a linear code, then a matrix $\bG$ over $\mathbb{Z}/p^s\mathbb{Z}$ is called a \textbf{generator matrix} for $\mathcal{C}$, if it has the code as row span and a matrix $\bH$ is called  a \textbf{parity-check matrix} for $\mathcal{C}$ if it has the code as kernel. 
\end{definition}
For a code  $\mathcal{C} \subseteq \left(\mathbb{Z}/p^s\mathbb{Z}\right)^n$,  we denote by $\mathcal{C}_S$ the code consisting of all codewords $\bc_S$, where $\bc \in \mathcal{C}.$
Also the notion of information set remains as in the classical case.
\begin{definition}
Let  $\mathcal{C} \subseteq \left(\mathbb{Z}/p^s\mathbb{Z}\right)^n$  be a linear code of rank $K$, then a set $I \subseteq \{1, \ldots, n\}$ of size $K$ is called an \textbf{information set} of $\mathcal{C}$ if 
$\mid \mathcal{C}_I \mid = \mid \mathcal{C} \mid. $
\end{definition}

In this paper we are interested in the Lee metric, which can be thought of as the $L_1$ norm modulo $p^s.$
\begin{definition}
Let $x \in \mathbb{Z}/p^s\mathbb{Z}$. The \textbf{Lee weight} of $x$ is given by \[\text{wt}_L(x) = \min\{x, \mid p^s -x \mid\}.\]
The Lee weight of a vector is then defined additively, i.e., for $\bx \in \left(\mathbb{Z}/p^s\mathbb{Z}\right)^n$, we have \[\text{wt}_L(\bx) = \sum_{i=1}^n \text{wt}_L(x_i).\]
Finally, this weight induces a distance, that is, for $\bx,\by \in \left(\mathbb{Z}/p^s\mathbb{Z}\right)^n$ the \textbf{Lee distance} between $\bx$ and $\by$ is given by $d_L(\bx,\by) = \text{wt}_L(\bx-\by).$
\end{definition}
Let us denote by $M \coloneqq \lfloor \frac{p^s}{2}\rfloor$, then one can easily see that for $x \in \mathbb{Z}/p^s\mathbb{Z}$ we have $0 \leq \text{wt}_L(x) \leq M.$
The Lee-metric ball, respectively the Lee-metric sphere of radius $r$ around $\bx \in \left(\mathbb{Z}/p^s\mathbb{Z}\right)^n$ are defined as 
\begin{align*} 
B(\bx,r,n,p^s) & = \{\by \in \left(\mathbb{Z}/p^s\mathbb{Z}\right)^n \mid d_L(\bx-\by) \leq r\}, \\
S(\bx,r,n,p^s) &= \{\by \in \left(\mathbb{Z}/p^s\mathbb{Z}\right)^n \mid d_L(\bx-\by) = r\}. \end{align*}
Since the size of a Lee-metric ball or a Lee-metric sphere is independent of the center, we will denote their cardinalities by 
\begin{align*}
    V(r,n,p^s) = \mid B(0,r,n,p^s) \mid, & \qquad
    F(r,n,p^s)  = \mid S(0,r,n,p^s) \mid.
\end{align*}

\begin{definition}
Let  $\mathcal{C} \subseteq \left(\mathbb{Z}/p^s\mathbb{Z}\right)^n$  be a linear code endowed with the Lee-metric, then the \textbf{minimum Lee distance} of $\mathcal{C}$ is given by 
\[d_L(\mathcal{C}) = \min\{ d_L(\bx,\by) \mid \bx \neq \by \in \mathcal{C}\}.\]
\end{definition}

Since in this paper we are interested in algorithms that have as input a code generated by a matrix chosen uniformly at random, due to the result in \cite[Proposition 16]{free}, we are allowed to assume that our code is free, i.e., $k=K$ and a generator matrix and a parity-check matrix have up to permutations of columns the following form 
\begin{align*}
    \bG = \begin{pmatrix} \text{Id}_k & \bA \end{pmatrix}, \ \bH = \begin{pmatrix} \text{Id}_{n-k} & \bB \end{pmatrix},
\end{align*} where $\bA \in \left( \mathbb{Z}/p^s \mathbb{Z}\right)^{k \times (n-k)}$ and $\bB \in \left(\mathbb{Z}/p^s\mathbb{Z}\right)^{(n-k) \times k}.$

In addition, in \cite[Theorem 20]{free} it was shown that such a random code also attains with high probability the Gilbert-Varshamov bound.

Let $AL(n,d,p^s)$ denote the maximal cardinality of a code $\mathcal{C} \subseteq (\mathbb{Z}/p^s\mathbb{Z})^n$ of minimum Lee distance $d$ and let us  consider the maximal information rate 
\[R(n,d,p^s) := \frac{1}{n} \log_{p^s}(AL(n,d,p^s)),\]
for $0 \leq d \leq n M$. We define the relative minimum distance to be  $\delta := \frac{d}{nM}.$

\begin{theorem}[Asymptotic Gilbert-Varshamov Bound \cite{leeas}]\label{asympt_GV}
It holds that
\[\liminf\limits_{n \to \infty}R(n,\delta Mn, p^s) \geq \lim\limits_{n \to \infty} \left( 1 - \frac{1}{n} \log_{p^s}(  V( \delta Mn, n,p^s)  ) \right).\]
\end{theorem}
We compute the asymptotic ball size, i.e., $\lim\limits_{n \to \infty}   \frac{1}{n} \log_{p^s}(  V( \delta Mn, n,p^s)  )$, in Section \ref{sec:spheres}.\\

Let  $\mathcal{C} \subseteq \left(\mathbb{Z}/p^s\mathbb{Z}\right)^n$  be a linear code with parity-check matrix $\bH, $ then for an $\bx \in \left(\mathbb{Z}/p^s\mathbb{Z}\right)^n$ we say that $\bs= \bx\bH^\top$ is a \textbf{syndrome}.
In this paper we give an algorithm that solves the following problem, called Lee syndrome decoding problem (LSDP), which was shown to be NP-complete in \cite{leenp}:
\begin{problem}\label{prob:LSDP}
 Let $\bH \in \left(\mathbb{Z}/p^s\mathbb{Z}\right)^{(n-k)\times n}, \bs \in \left(\mathbb{Z}/p^s\mathbb{Z}\right)^{n-k}$ and $t \in \mathbb{N}.$ Find $\be \in \left(\mathbb{Z}/p^s\mathbb{Z}\right)^n$ such that $\bs=\be\bH^\top$ and $\text{wt}_L(\be)=t.$  
\end{problem}
 
 To this end, we assume that the input parity-check matrix $\bH$ is chosen uniformly at random in $\left(\mathbb{Z}/p^s\mathbb{Z}\right)^{(n-k)\times n}$ and that there exists a solution  $\be \in \left(\mathbb{Z}/p^s\mathbb{Z}\right)^n$, which was chosen uniformly at random in $S(0,t,n,p^s)$ and set $\bs$ to be its syndrome $\bs= \be\bH^\top.$ 
 We provide two new algorithms, taking care of two different scenarios. In the first scenario, we want to decode up to the minimum distance of the code having $\bH$ as parity-check matrix.  For this, we let  $d_L$ be the minimum distance from the Gilbert-Varshamov bound, then even if we assume full distance decoding, i.e., $t=d_L$,  we expect to have a unique solution $\be$ to Problem \ref{prob:LSDP}. In fact, the expected number of solutions to the LSDP is given by 
 \[N=\frac{F(t,n,p^s) }{p^{s(n-k)}} = \frac{F(d_L,n,p^s) }{p^{s(n-k)}} \leq 1. \]

In the second scenario, we consider a Lee weight $t$ which is  beyond the minimum distance, and solve this new problem by reversing the idea of the first algorithm. The main idea of these new algorithms is to use the results of \cite{distribution}, which provide us with additional information on the unique solution $\be \in \left(\mathbb{Z}/p^s\mathbb{Z}\right)^n$. For example, the expected number of entries of $\be$ having a fixed Lee weight.

\subsection{Asymptotics of Lee Spheres} \label{sec:spheres}
In the complexity analysis of our algorithm, we are interested in the asymptotic size of Lee spheres, Lee balls, and some types of restricted Lee spheres. All these quantities can be described using generating functions, and their limit for $n$ going to infinity can be computed using the saddle point technique used in \cite{saddle}.

Let $\Phi(x) = f(x)^n g(x)$ be a generating function, where $f(x)$ and $g(x)$ do not depend on $n$.  Let us denote by $[x^t]\Phi(x)$ the coefficient of $x^t$ in $\Phi(x).$ We want to estimate this coefficient for $t = Tn$ for some fixed $T \in (0,1)$.

\begin{lemma}[\text{\cite[Corollary 1]{saddle}}]\label{lemGardySole}
Let $\Phi (x) = f(x)^n g(x)$ with $f(0)\neq 0$, and $t(n)$ be a function in $n$. Set $T := \lim_{n\rightarrow \infty}t(n)/n$ and set $\rho$ to be the solution to 
\[ \Delta(x):= \frac{x f'(x)}{f(x)} = T .\]
If $ \Delta'(\rho) >0$, and the modulus of any singularity of $g(x)$ is larger than $\rho$, then for large $n$
\[ \frac{1}{n} \log_{p^s}( [x^{t(n)}]\Phi(x)) \approx \log_{p^s}(f(\rho)) - T \log_{p^s}(\rho) + o(1) .\]
\end{lemma}

\subsubsection{(Restricted) Lee spheres}
The generating functions of the sizes of Lee spheres and Lee balls are known to be $\Phi(x) = f(x)^n$ and $\Phi'(x) = \frac{f(x)^n}{1-x}$, respectively, where \[f(x):=\left\{\begin{array}{ll} 1+2\sum_{i=1}^M x^i & \text{ if } p \neq 2, \\ 1+2\sum_{i=1}^{M-1} x^i + x^M & \text{ if } p=2. \end{array}\right.  \]
 It clearly follows that $F(t,n,p^s) = [x^t] \Phi(x)$ and $V(t,n,p^s) = [x^t] \Phi'(x)$. 
 The asymptotics of these sizes have been computed in \cite{saddle,leenp}. 
 
 In general, we can also compute the sizes of the restricted Lee spheres,   where each entry has the Lee weight smaller than $r$, respectively larger than $r$, for some $r \in \{0,\ldots,M\}$ 
 \begin{align*}
 F_{(r)}(t, n, p^s) &= \mid \ \{\bx \in \{0,  \pm 1, \ldots, \pm r\}^n \mid \text{wt}_L(\bx)=t \ \text{in} \ \mathbb{Z}/p^s\mathbb{Z}\} \ \mid,  \\
 F^{(r)}( t, n, p^s) &= \mid \ \{\bx \in \{  \pm r, \ldots, \pm M\}^n \mid \text{wt}_L(\bx) = t \ \text{in} \ \mathbb{Z}/p^s\mathbb{Z}\} \ \mid. \end{align*}
 
 The generating function of the size of the restricted Lee sphere $F_{(r)}(t, n, p^s)$ is given by $\Phi_{(r)}(x) = f_{(r)}(x)^n$, 
 where \[f_{(r)}(x) := \left\{\begin{array}{ll}  1+2\sum_{i=1}^{M-1} x^i + x^M & \text{ if } p=2 \ \text{and} \ r =M, \\ 1+2\sum_{i=1}^r x^i &  \text{ otherwise}. \end{array}\right.  \]

Whereas, for 
 $F^{(r)}(t,n,p^s)$ the generating function is given by  $\Phi^{(r)}(x)=f^{(r)}(x)^n,$ where
\begin{align*} 
f^{(r)}(x) = \begin{cases} f_{(M)}(x) &  \text{ if } r=0, \\ 2\sum_{i=r}^{M-1} x^i + x^M & \text{ if } p=2 \text{ and } r>0, \\ 2\sum_{i=r}^{M} x^i  & \text{ if } p\neq 2 \text{ and } r>0. \end{cases}
\end{align*}
Note that the coefficient of $x^t$ in $\Phi^{(r)}(x)$ is equal to the coefficient of $x^{t-rn}$ in $\Psi^{(r)}(x) = g^{(r)}(x)^n$, where 
\[g^{(r)}(x) := \left\{\begin{array}{ll}  f_{(M)}(x) &  \text{ if } r=0, \\ 2\sum_{i=0}^{M-1-r} x^i + x^{M-r} & \text{ if } p=2 \text{ and } r>0, \\ 2\sum_{i=0}^{M-r} x^i  & \text{ if } p\neq 2 \text{ and } r>0. \end{array}\right.  \]
 
In particular, we have that $F_{(r)}(t,n,p^s) = [x^t]\Phi_{(r)}(x)$ and $F^{(r)}(t,n,p^s) = [x^{t-rn}]\Psi^{(r)}(x).$ Note that $F_{(M)}(t,n,p^s) = F^{(0)}(t,n,p^s) = F(t,n,p^s)$. Using Lemma \ref{lemGardySole}, we get the following asymptotic behavior of restricted Lee spheres.

\begin{corollary}\label{asympt_V} Let $T \in [0,M)$ and $t = t(n)$ be a function of $n$ such that $t(n) := Tn$ for large $n$. Then,
\begin{enumerate}
    \item for $p \neq 2$ or $r < M$, we get
    \[\lim\limits_{n \to \infty}  \frac{1}{n} \log_{p^s}( F_{(r)}(t(n),n, p^s))=\log_{p^s}(f_{(r)}(\rho)) - T \log_{p^s}(\rho),\]
    where $\rho$ is the unique real positive solution of $2 \sum_{i=1}^r (i-T) x^i = T$ and \[f_{(r)}(\rho)= 1 + 2\sum_{i=1}^r \rho^i = \frac{r(\rho+1)+1}{(1-\rho)(r-T)+1},\]
    \item for $p=2$ and $r=M$, respectively $r'=0$, we get
    \begin{align*}
    \lim\limits_{n \to \infty}  \frac{1}{n} \log_{p^s}( F_{(r)}( t(n), n,p^s)) = \log_{p^s}(f_{(r)}(\rho)) - T \log_{p^s}(\rho),
    \end{align*}
    where $\rho$ is the unique real positive solution of $2 \sum_{i=1}^{M-1} (i-T) x^i + (M-T) x^M = T$ and 
    \begin{align*}
       g^{(r')}(\rho) & = f_{(r)}(\rho)  = 1 + 2 \sum_{i=1}^{M-1} {\rho}^i + {\rho}^M  \\& = \frac{{\rho}^{M+1}(T-M) + {\rho}^M(T-M+1) + {\rho}(T-M) + T +M +1}{{\rho}(T-M) + M+1-T},
    \end{align*}

    \item for $p=2$ and $0<r \leq T$, we get 
    \[\lim\limits_{n \to \infty}  \frac{1}{n} \log_{p^s}( F^{(r)}(t(n),n, p^s))=\log_{p^s}(g^{(r)}(\rho)) - (T-r) \log_{p^s}(\rho),\]
    where $\rho$ is the unique real positive solution of 
    \[2\sum_{i=1}^{M-1-r} (i-T+r) x^i + (M-T) x^{M-r}=2(T-r),\]
    and
    \[g^{(r)}(\rho)= 2\sum_{i=0}^{M-1-r} \rho^i + \rho^{M-r} = \frac{\rho^{M-r+1}+\rho^{M-r} -2}{\rho-1},\]
    
        \item for $p\neq 2$ and $0<r\leq T$, we get 
    \[\lim\limits_{n \to \infty}  \frac{1}{n} \log_{p^s}( F^{(r)}(t(n),n, p^s))=\log_{p^s}(g^{(r)}(\rho)) - (T-r) \log_{p^s}(\rho),\]
    where $\rho$ is the unique real positive solution of $2\sum_{i=1}^{M-r} (i-T+r)x^i= 2(T-r) $  and
    \[ g^{(r)}(\rho)= 2 \sum_{i=0}^{M-r} \rho^i =  \frac{2\rho^{M-r+1}-2}{\rho-1} .\]
\end{enumerate}
\end{corollary}
\begin{proof}
For the parts 1 and 2, we apply Lemma \ref{lemGardySole} to the generating function $\Phi_{(r)}(x)$ and obtain the mentioned results, similar to $r=M$ for $f_{(r)}$ case proved in \cite[Lemma 2.6]{leenp}. For the parts 3 and 4, we 
apply Lemma \ref{lemGardySole} to the generating function $\Psi^{(r)}(x)$ and obtain the mentioned results. 
 
\end{proof}

\begin{remark}\label{rem} Note that, for $p$ odd (respectively, even), we get $T \geq M(M+1)/(2M+1)$ (respectively, $T \geq M/2$) if and only if \[ \lim\limits_{n \to \infty}  \frac{1}{n} \log_{p^s}(  V( Tn, n,p^s)  ) =1.\]
Hence, if $0<R$, then a code that attains the asymptotic Gilbert-Varshamov bound has 
\[ \lim\limits_{n \to \infty}  \frac{1}{n} \log_{p^s}(  V( T n, n,p^s)  ) = 1-R <1,\]
and we immediately get that $T < M(M+1)/(2M+1)$ if $p$ is odd, or $T < M/2$ if $p$ is even. 
\end{remark}

\subsubsection{Restricted Compositions}

Let us denote by $C(v,t, \lambda,n,p^s)$ the number of weak compositions $\pi$ of $v$, which fit into the composition $\lambda$ of $t$ and both having $n$ part sizes. That is, the maximal part sizes of $\pi$ are given by $\lambda$, i.e.,  for all $i \in \{1, \ldots, n\}$ we have $\pi_i \leq \lambda_i$. Equivalently, the Young-Tableau of $\pi$ fits into the Young-Tableau of $\lambda$.
In addition, we have that $\lambda$ is a composition which has $n$ parts. The reason we are interested in this number is that we can think of the Lee weight 
 composition of a vector $\be \in \left(\mathbb{Z}/p^s\right)^n$ as
   $\lambda=(\lambda_1, \ldots, \lambda_{n}),$
   which is such that $\lambda_i=\text{wt}_L((\be)_i)$.
Let $V= \lim\limits_{n \to \infty} v(n)/n$. Then for $m= \max\{\lambda_i \mid i \in \{1, \ldots, n\}\}$ the number of weak compositions which fit into $\lambda$ have the
generating function \[\Phi(z) = \prod_{i=1}^n \left( \sum_{j=0}^{\lambda_i} z^j \right) =\prod_{i=1}^{m} \left(\sum_{j=0}^i z^j\right)^{c_i n},\] where  $c_i$  corresponds to the multiplicity of $i$ in the composition $\lambda$, i.e., there are $c_in$ entries of $\be \in \left( \mathbb{Z}/p^s\mathbb{Z}\right)^n$ which have Lee weight $i$. Thus $\Phi(z)=f(z)^n,$ for \[f(z)=\prod_{i=1}^{m} \left(\sum_{j=0}^i z^j\right)^{c_i}.\]  To get the asymptotics of $C(v,t, \lambda,n,p^s)$ we are interested in the coefficient of $z^v$ in $\Phi(z)$. Now using the saddle point technique of \cite{saddle} we define $\Delta(f(z))= \frac{zf'(z)}{f(z)}.$ Let $\rho$ be the unique positive real solution to $\Delta(f(z))=V.$ Then 
\[\lim\limits_{n \to \infty} \frac{1}{n} \log_{p^s} \left( C(v,t, \lambda,n,p^s) \right) = \log_{p^s}(f(\rho)) - V\log_{p^s}(\rho).\]

\begin{lemma}\label{asympt_smallballs} 
Let us consider a weak composition $\lambda=(\lambda_1, \ldots, \lambda_n)$ of $t$ with $\lim\limits_{n \to \infty} t(n)/n=T.$ In addition, let us consider a positive integer $v \leq t$ with $\lim\limits_{n \to \infty} v(n)/n=V.$ Let $m=\max\{\lambda_i \mid i \in \{1, \ldots, n\}\}.$
 If  $0 \leq V <M$, then
\[\lim\limits_{n \to \infty}  \frac{1}{n} \log_{p^s}(  C(v,t, \lambda,n,p^s) )=\log_{p^s}(f(\rho)) - V \log_{p^s}(\rho),\]
where $\rho$ is the unique real positive solution of \[\sum_{i=1}^m c_i \dfrac{z+2z^2+\cdots+iz^i}{1+z+\cdots+z^i} = V.\]

\end{lemma} 
 
\section{Distribution of a Random Lee Vector}\label{sec:distr}
In this section, we analyze the error vector $\be$ that is chosen uniformly at random from $S(0,t,n,p^s).$ We first recall the results from \cite{distribution} that studies the distribution of the entries of a random vector having a fixed Lee weight.

Let $E$ denote a random variable corresponding to the realization of an entry of $\be$. As $n$ tends to infinity we have the following result on the distribution of the elements in $\be$.
\begin{lemma}[{\cite[Lemma 1]{distribution}}]\label{lem:marginal}
    For any $j \in \zps$,
    \begin{align}\label{equ:marginal}
        \prob (E = j) = \frac{1}{Z(\beta)} \exp(-\beta \lweight{j}),
    \end{align}
    where $Z$ denotes the normalization constant and $\beta$ is the unique real solution to the constraint $t/n = \sum_{i = 0}^{p^s-1} \lweight{i} \prob(E = i)$.
\end{lemma}
Note, that if $\beta = 0$, the entries of $\be$ are uniformly distributed over $\zps$. In that case, the relative weight of a randomly chosen entry is equal to $\frac{p^{2s}-1}{4p^s}$ if $p$ is odd, respectively $p^s/4$  if $p = 2$. Furthermore, if $\beta > 0$ the relative weight becomes smaller. In addition, since the marginal distribution \eqref{equ:marginal} is an exponential function with negative exponent, it is decreasing in the weight. This means that for $\beta > 0$ the elements of smallest Lee weight, i.e., $0$, are the most probable, then elements of weight $1$ until the least probable Lee weight $M$. Let us emphasize here that, by Remark \ref{rem}, if we are in scenario 1, where we decode up to the minimum distance given by the Gilbert-Varshamov bound, we will always have $t/n \leq M/2$, which is roughly the threshold $\frac{p^{2s}-1}{4p^s}$ for $p$ odd, respectively  $p^s/4$ for $p = 2$, and hence elements of small weight will always be more probable. In the other case, where $\beta<0$, the elements of largest Lee weight, i.e., $M$, are the most probable, followed by the elements of weight $M-1$, and so on, until the least probable of Lee weight $0$. This is the case for the second scenario, where we decode  beyond the minimum distance, i.e., $t/n \geq M/2.$
\begin{figure}
    \centering
    \begin{tikzpicture}[scale=0.7]
	\begin{axis}[
	width=18cm,
	height=15cm,
	legend pos=north east,
	grid=both,
	grid style={dotted,gray!50},
	xmin=-23,
	xmax=23,
	ymin=0,
	ymax=0.2,
	ylabel={Probability},
	xlabel={Elements of $\mathbb{Z}/47\mathbb{Z}$},
	y label style={at={(0.0,0.5)}},
	legend cell align=left,
	legend style={font=\footnotesize},
	]
	
	\addplot[blue!70!cyan,solid,mark=o,mark options={solid}] table[x=n,y=prob] {./dists_q47_delta3.0.txt};\addlegendentry{$T = 3$}
	\addplot[olive!70!green,solid,mark=+,mark options={solid}] table[x=n,y=prob] {./dists_q47_delta8.0.txt};\addlegendentry{$T = 8$}
	\addplot[orange!70!yellow,solid,mark=diamond,mark options={solid}] table[x=n,y=prob] {./dists_q47_delta11.7447.txt};\addlegendentry{$T = \frac{(p^s)^2-1}{4p^s} \approx 11.7447$}
	\addplot[violet!70!black,solid,mark=star,mark options={solid}] table[x=n,y=prob] {./dists_q47_delta16.0.txt};\addlegendentry{$T = 16$}

	\end{axis}
\end{tikzpicture}
    \caption{Marginal distribution for the elements in $\mathbb{Z}/47\mathbb{Z}$ for different values of  $T = \lim_{n \to \infty} t(n)/n$.}
    \label{fig:marg_dist}
\end{figure}
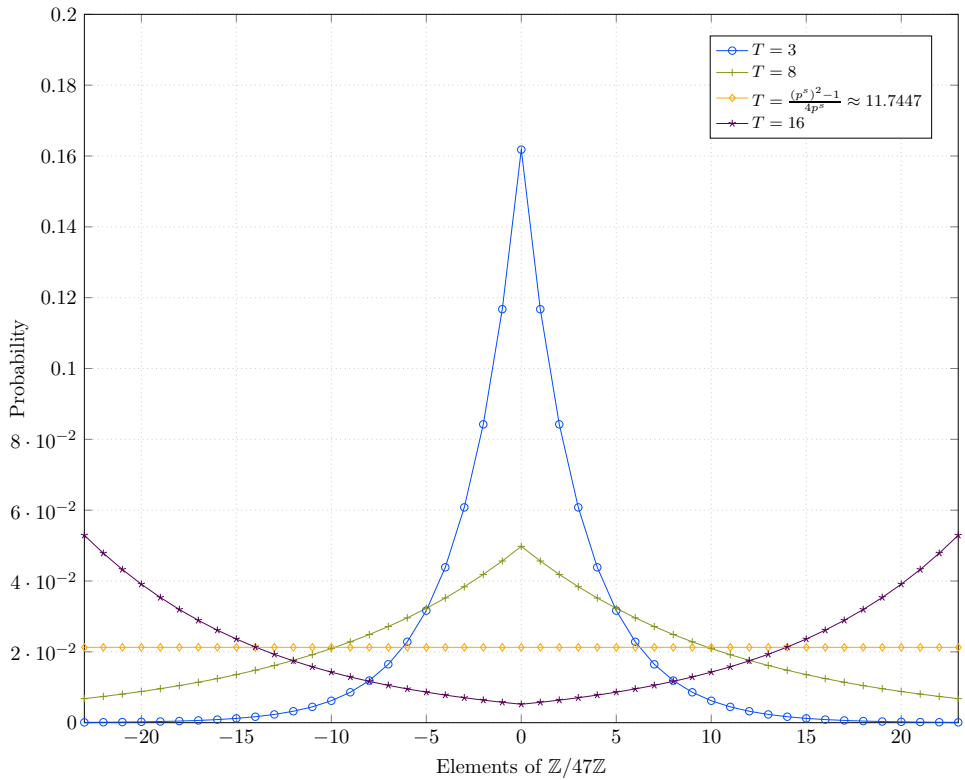
As a direct consequence of Lemma \ref{lem:marginal}, we can give the probability of a random entry $E$ having some given Lee weight
\begin{align}\label{equ:prob_LW}
    \prob (\lweight{E} = j) 
        =
        \begin{cases}
            \prob( E = j) & \text{if } (j = 0) \text{ or } (j = M \text{ and } q \text{ is even}), \\
            2 \prob( E = j) & \text{else}. 
        \end{cases}
\end{align}

In this work, we are interested in the expected number of entries that have 'large' Lee weight, i.e., entries having Lee weight larger than a threshold $r \in \{1,\ldots,M-1\}$. 
Let $\psi(r,t,n,p^s)$ denote the expected number of entries of $\be$ which have Lee weight larger than $r$ and let $\varphi(r,t,n,p^s)$ denote the expected Lee weight of $\be$ without the entries of larger Lee weight than $r$. In addition, for some randomly chosen set $S \subseteq \{1, \ldots, n\}$ of size  $0  \leq \ell \leq n$, let us denote by $\sigma(\ell,t,n,p^s)$ the expected support size of $\be_S$. 

\begin{lemma} Let $\be$ be chosen  uniformly at random  in $S(0,t,n,p^s)$,  $r \in \{0,\ldots,M\}$ and $0 \leq \ell \leq n$. Then
\begin{align*}
\psi(r,t,n,p^s) & = n \sum_{i=r+1}^M \prob (\lweight{E} = i), \\
\varphi(r,t,n,p^s) & = n \sum_{i=0}^{r} i \cdot \prob(\lweight{E}=i), \\
\sigma(\ell,t,n,p^s) & = \ell \sum_{i=1}^M \prob (\lweight{E} = i).
\end{align*}
\end{lemma}
\begin{proof}
The proof easily follows from \eqref{equ:prob_LW} and using the assumption that each entry of $\be$ is independent.  
\end{proof}

\section{Restricted-Balls Algorithm}\label{sec:algo}

The idea of the new information set decoding algorithms is to use the information on the uniformly chosen  $\be \in S(0,t,n,p^s)$. 
We start with the algorithm for the first scenario, where we only decode up to the minimum distance given by the Gilbert-Varshamov bound and later adapt this algorithm to the second scenario, where we decode  beyond the minimum distance.

\subsection{Decoding up to the Minimum Lee Distance}

The high level idea lies in the following observation: for $t/n< M/2$, which we have due to the Gilbert-Varshamov bound, we know as $n$ grows large that $0$ is the most likely entry of $\be$, the second most likely is $\pm 1$ and so on, until the least likely entry is $\pm M$. Hence, if we define a threshold Lee weight $0\leq r \leq M$, then with a high probability (depending on the choice $r$) we have that all entries of $\be$ of Lee weight larger than $r$ can be found outside an information set. Thus, using the partial Gaussian elimination (PGE) algorithms, we are left with finding a smaller error vector, which only takes values in $\{0, \pm 1, \ldots, \pm r\}.$ This will make a huge difference for algorithms such as the Lee-metric BJMM \cite{leenp}, where the list sizes are the main factor in the cost and these can now be immensely reduced.

In general, this idea can be considered as a framework, where one can apply any algorithm that solves the smaller instance, but now in a smaller space.
The framework takes as input $(\bH, \bs, t, r, \mathcal{S})$, where $\mathcal{S}$ denotes a solver for the smaller instance in the space $\{0, \pm 1, \ldots, \pm r\}$, which instead of outputting a list of possible solutions for the smaller instance immediately checks whether the smaller solution at hand leads to a solution of the original instance. More precisely, the framework on $(\bH, \bs, t, r, \mathcal{S})$ works as follows:\\

Let us consider an instance of the LSDP, given by $\bH \in \left(\mathbb{Z}/p^s\mathbb{Z}\right)^{(n-k)\times n}$, $\bs \in \left(\mathbb{Z}/p^s\mathbb{Z}\right)^{n-k}$ and $ t \in \mathbb{N}$, with $t/n<M/2$.

\noindent \textit{Step 1:}\ For some $0 \leq \ell \leq n-k$, we will bring the parity-check matrix into partial systematic form by multiplying $\bH$ with some invertible $\bU \in \left(\mathbb{Z}/p^s\mathbb{Z}\right)^{(n-k)\times (n-k)}$ and adapting the syndrome accordingly to $\bs'=\bs\bU^\top$. For simplicity, assume that we have an information set in the last $k$ positions. Thus,  the LSDP becomes
\[\begin{pmatrix} \be_1 & \be_2 \end{pmatrix} \begin{pmatrix} \text{Id}_{n-k-\ell} & 0 \\ \bA^\top & \bB^\top \end{pmatrix}   = \begin{pmatrix} \bs_1 & \bs_2 \end{pmatrix}, \] where $\bA  \in \left(\mathbb{Z}/p^s\mathbb{Z}\right)^{(n-k-\ell) \times (k+\ell)}, \bB \in \left( \mathbb{Z}/p^s\mathbb{Z}\right)^{\ell \times (k+\ell)}, \bs_1 \in \left( \mathbb{Z}/p^s\mathbb{Z}\right)^{n-k-\ell}$ and $\bs_2 \in \left(\mathbb{Z}/p^s\mathbb{Z}\right)^\ell.$
Thus, we have to solve two parity-check equations:
\begin{align}
    \be_1+\be_2 \bA^\top &= \bs_1, \notag \\
    \be_2 \bB^\top &= \bs_2, \label{seceq}
\end{align}
where we assume that $\be_2$ has Lee weight $v$ and $\be_1$ has Lee weight $t-v,$ for some positive integer $0 \leq v \leq t.$

\noindent \textit{Step 2:} \ We solve the smaller instance of the LSDP given by Equation \eqref{seceq} using algorithm $\mathcal{S}$. In particular, we find an error vector $\be_2$ such that $ \be_2 \bB^\top = \bs_2$, $\lweight{\be_2} = v$, and it has entries in $\{0,\pm 1, \ldots, \pm r\}$. Instead of storing a list of solutions $\be_2$, $\mathcal{S}$ will immediately check whether $\be_1=\bs_1-\be_2\bA^\top$ has the remaining Lee weight $t-v.$ Clearly, $v$ will also depend on the choice of $r$.\\

Solving the smaller instance can be achieved using various techniques, for example via Wagner's approach used in \cite{leenp,thomas} or via the representation technique used in \cite{leenp}. However, we have to slightly adapt these techniques to make use of the assumption that the entries are restricted to $\{0,\pm 1, \ldots, \pm r\}$.

Let $S_{(r)}(0,v,n, p^s)$ denote the Lee sphere of weight $v$ centered at the origin with entries restricted to $\{0,\pm 1,\ldots,\pm r\}$, i.e., \[ S_{(r)}(0,v,n,p^s) := \{ \bx \in \{0,\pm 1, \ldots,\pm r\}^n | \lweight{\bx} = v \ \}. \] In the following lemma, we show that if $\be$ is a random vector of length $n$ and Lee weight $t$ which splits as $(\be_1,\be_2)$ with $\be_2 \in S_{(r)}(0,v,k+\ell,p^s)$, then $\be_2$ has a uniform distribution in $S_{(r)}(0,v,k+\ell,p^s)$.
\begin{lemma} \label{lemma:e2Distribution}
Let $\be$ be chosen uniformly at random  in $S(0,t,n,p^s)$ such that $\be = (\be_1, \be_2)$ with $\be_2 \in S_{(r)}(0,v,k+\ell,p^s)$. Then $\be_2$ follows a uniform distribution in $S_{(r)}(0,v,k+\ell,p^s)$, and henceforth $\be_1$ follows a uniform distribution in $S(0,t-v,n-k-\ell,p^s)$.
\end{lemma}
\begin{proof}
We note that for an arbitrary $\be_2 \in S_{(r)}(0,v,k+\ell,p^s)$, there are exactly $|S(0,t-v,n-k-\ell,p^s)|$ possible $\be$ that restrict to $\be_2$ in their last $k+\ell$ coordinates.  
Therefore, if $\be$ is chosen uniformly at random, then each $\be_2$ has an equal chance of being chosen in $S_{(r)}(0,v,k+\ell,p^s)$.  
\end{proof}
As a corollary, we see that this splitting of $\be$ comes with a probability of 
$$P= F_{(r)}(v,k+\ell, p^s)F(t-v, n-k-\ell,p^s) F(t,n,p^s)^{-1}.$$
\noindent \textit{Using the BJMM-Approach:} \  Let us consider an adaption of the Lee-BJMM algorithm from \cite{leenp}, where two levels were the optimal choice and proved to remain the optimal choice also for this new algorithm.
Although the smaller error vector $\be_2$ now only has entries in $\{0, \pm 1, \ldots, \pm r\},$   to enable representation technique, we will assume that such a vector $\be_2$ is built from the sum of two vectors $\by_1+\by_2$, where $\varepsilon$ many of their positions cancel out and thus are allowed to live in the whole ring $\mathbb{Z}/p^s\mathbb{Z}.$ Let us denote these positions by $\mathcal{E}.$\\

The high level idea of BJMM on two levels is as follows: we  split $\be_2$ as
\begin{align*}
    \be_2 &= \by_1 +\by_2 \\
        &= (\bx_1^{(1)}, \bx_2^{(1)}) + (\bx_1^{(2)}, \bx_2^{(2)}).
\end{align*}
Thus, for the syndrome equation to be satisfied, we want that 
$$\bs_2 = \be_2 \bB^\top = \by_1\bB^\top +  \by_2\bB^\top.$$
Let us also split $\bB \in \left(\mathbb{Z}/p^s\mathbb{Z}\right)^{\ell \times (k+\ell)}$ into two matrices 
$\bB= \begin{pmatrix} \bB_1 & \bB_2 \end{pmatrix},$ where $\bB_i \in \left(\mathbb{Z}/p^s\mathbb{Z}\right)^{\ell \times (k+\ell)/2},$ for $i \in \{1,2\}.$ Then in a first merge to get $\by_i = (\bx_1^{(i)}, \bx_2^{(i)})$ we want for $i=1$, that they give the syndrome 0, i.e., 
\begin{align*}
    \bx_1^{(1)} \bB_1^\top= - \bx_2^{(1)} \bB_2^\top,
\end{align*}
and for $i=2$ that they give the syndrome $\bs_2$, i.e., 
\begin{align*}
    \bx_1^{(2)} \bB_1^\top=\bs_2 - \bx_2^{(2)} \bB_2^\top.
\end{align*}
Let us split $\mathcal{E}$ evenly into two disjoint index sets, i.e., $\mathcal{E} = \mathcal{E}_1 \cup \mathcal{E}_2$ such that $\mid \mathcal{E}_1 \mid = \mid \mathcal{E}_2\mid$ and $\mathcal{E}_1 \cap \mathcal{E}_2 = \emptyset$.
The base lists $\mathcal{B}_i$ for $i \in \{1, 2\}$ are then be built as follows
\begin{align*}
    \mathcal{B}_i = \{\nu(\bx) \mid \bx_{\mathcal{E}_i^C}  \in  \{0, \ldots, \pm r\}^{(k+\ell-\varepsilon)/2}, \bx_{\mathcal{E}_i}  \in   \left(\mathbb{Z}/p^s\mathbb{Z}\right)^{\varepsilon/2}  , \text{wt}_L(\bx_{\mathcal{E}_i^C})=v/4, \nu  \in  S_{(k+\ell)/2}\}.
\end{align*}

For some positive integer $u\leq n$ and $\bx,\by \in \left(\mathbb{Z}/p^s\mathbb{Z}\right)^n$, we write $\bx=_u \mathbf{y}$, to denote that $\bx=\mathbf{y}$ in the last $u$ positions. Let us define the following two sets.
\begin{align*} 
    \mathcal{L}_1 & = \{\mu(\by) \mid \by_{\mathcal{E}^C}  \in  \{0,  \ldots, \pm r\}^{k+\ell-\varepsilon}, \by_{\mathcal{E}}  \in  \left(\mathbb{Z}/p^s\mathbb{Z}\right)^\varepsilon , \text{wt}_L(\by_{\mathcal{E}^C}) = v/2, \by \bB^\top  =_u  0, \mu  \in  S_{k+\ell}. \}, \\
    \mathcal{L}_2 & = \{\mu'(\by) \mid \by_{\mathcal{E}^C}  \in  \{0,  \ldots, \pm r\}^{k+\ell-\varepsilon}, \by_{\mathcal{E}}  \in  \left(\mathbb{Z}/p^s\mathbb{Z}\right)^\varepsilon , \text{wt}_L(\by_{\mathcal{E}^C}) = v/2, \by \bB^\top  =_u   \bs_2,  \mu'  \in  S_{k+\ell} \}.
\end{align*} 
 
Performing a concatenation merge, we compute $\by_i=(\bx_1^{(i)}, \bx_2^{(i)})$ for $(\bx_1^{(i)}, \bx_2^{(i)}) \in \mathcal{B}_1 \times \mathcal{B}_2$ on the syndromes 0 and $\bs_2$ and $u$ positions.
Hence, to get $\by_1 \in \mathcal{L}_1$, we merge $\by_1 = (\bx_1^{(1)}, \bx_2^{(1)})$, such that $$ \bx_1^{(1)} \bB_1^\top =_u - \bx_2^{(1)}\bB_2^\top,$$ and to get $\by_2 \in \mathcal{L}_2$, we merge $\by_2 = (\bx_1^{(2)}, \bx_2^{(2)})$, such that $$\bx_1^{(2)}\bB_1^\top=_u \bs_2 - \bx_2^{(2)} \bB_2^\top.$$

We then merge $\mathcal{L}_1 \bowtie \mathcal{L}_2$ on the syndrome $\bs_2$ and $\ell$ positions, computing $\be_2=\by_1+\by_2$,  for $(\by_1, \by_2) \in \mathcal{L}_1 \times \mathcal{L}_2$
such that the positions $\mathcal{E}$ of $\by_1$ and $\by_2$ cancel out, i.e.,
$ {\by_1}_{\mathcal{E}}+{\by_2}_{\mathcal{E}}=0$
and $\text{wt}_L(\be_2)=v.$
\begin{remark}
Note that our base lists, as well as the lists $\mathcal{L}_i$ employ a permutation. Hence, it might happen that  the $\mathcal{E}$ positions are not equal for $\by_1$ and $\by_2$, and these positions might not cancel out. However, the algorithm will still succeed, 
since we will check within the merge, that $\by_i \in \mathcal{L}_i$ have the correct weight $v$. 
The only implication for the workfactor is that  the success probability in this case would even be larger, thus we are giving an upper bound on the cost.
\end{remark}
\begin{figure}
\begin{center}
    \begin{tikzpicture}[line cap=round,line join=round,>=triangle 45,x=1cm,y=1cm]
        \draw (-1.1, 3.75) node[left] {$\mathbf{e}_2 $};
        \draw[line width=1pt] (-1,4) -- (-1, 3.5) -- (11,3.5) -- (11,4) -- cycle;
        \draw[stealth-stealth] (-1, 4.1) -- (11, 4.1);
        \draw (5, 4.1) node[above] {$k+ \ell$};
        \fill[line width=1pt,fill = gray, fill opacity=0.3] (2,3.5) -- (8,3.5) -- (8,4) -- (2,4) -- cycle;
        \draw (5, 3.6) node[below] {$\underbrace{\hspace{6cm}}$};
        \draw (5, 3.3) node[below] {$\supp (\mathbf{e}_2) \in \set{\pm 1, \dots, \pm r}^{\card{\supp (\mathbf{e}_2)}}$};
        \draw (5, 3.75) node[] {$v$};
    
        \draw (-1.1,1.75) node[left] {$\mathbf{y}_1$};
        \draw[line width=1pt] (-1,1.5) -- (-1,2) -- (11,2) -- (11,1.5) -- cycle;
        \draw[line width = 0.75pt, style = dashed] (5, 2.15) -- (5, 1.25);
        
        \fill[line width=1pt,fill = gray, fill opacity=0.3] (2,2) -- (3.5,2) -- (3.5,1.5) -- (2,1.5) -- cycle;
        \draw (2.75,1.75) node[] {\footnotesize $v/4$};
        \draw[pattern = north east lines, pattern color = lightgray] (0, 2) rectangle (0.8, 1.5);
        \draw[stealth-stealth] (0, 2.1) -- (0.8, 2.1);
        \draw (0.4, 2.1) node[above] {\footnotesize $\varepsilon/2$};
        \draw (2, 1.6) node[below] {$\footnotesize\underbrace{\hspace{6cm}}$};
        \draw (2, 1.3) node[below] {\footnotesize $\mathbf{x}_1^{(1)}$};
        
        \fill[line width=1pt,fill = gray, fill opacity=0.3] (5,2) -- (6.5,2) -- (6.5,1.5) -- (5,1.5) -- cycle;
        \draw (5.75, 1.75) node[] {\footnotesize $v/4$};
        \draw[pattern = north east lines, pattern color = lightgray] (8.5, 2) rectangle (9.3, 1.5);
        \draw[stealth-stealth] (8.5, 2.1) -- (9.3, 2.1);
        \draw (8.9, 2.1) node[above] {\footnotesize $\varepsilon/2$};
        \draw (8, 1.6) node[below] {$\underbrace{\hspace{6cm}}$};
        \draw (8, 1.3) node[below] {\footnotesize $\mathbf{x}_2^{(1)}$};
        
        \draw (-1.1,0.25) node[left] {$\mathbf{y}_2 $}; 
        \draw[line width=0.75pt] (-1,0.5) -- (-1,0) -- (11,0) -- (11,0.5) -- cycle;
        \draw[line width = 0.75pt, style = dashed] (5, 0.65) -- (5, -0.25);
    
        \fill[line width=1pt,fill = gray, fill opacity=0.3] (5,0) -- (3.5,0) -- (3.5,0.5) -- (5,0.5) -- cycle;
        \draw (4.25, 0.25) node[] {\footnotesize $v/4$};
        \draw[pattern = north east lines, pattern color = lightgray] (0, 0) rectangle (0.8, 0.5);
        \draw[stealth-stealth] (0, 0.6) -- (0.8, 0.6);
        \draw (0.4, 0.6) node[above] {\footnotesize $\varepsilon/2$};
        \draw (2, 0.1) node[below] {$\underbrace{\hspace{6cm}}$};
        \draw (2, -0.2) node[below] {\footnotesize $\mathbf{x}_1^{(2)}$};
        
        \fill[line width=1pt,fill = gray, fill opacity=0.3] (8,0) -- (6.5,0) -- (6.5,0.5) -- (8,0.5) -- cycle;
        \draw (7.25,0.25) node[] {\footnotesize $v/4$};
        \draw[pattern = north east lines, pattern color = lightgray] (8.5, 0) rectangle (9.3, 0.5);
        \draw[stealth-stealth] (8.5, 0.6) -- (9.3, 0.6);
        \draw (8.9, 0.6) node[above] {\footnotesize $\varepsilon/2$};
        \draw (8, 0.1) node[below] {$\underbrace{\hspace{6cm}}$};
        \draw (8, -0.2) node[below] {\footnotesize $\mathbf{x}_2^{(2)}$};
    \end{tikzpicture}
\end{center}
\caption{Illustration of two levels decomposition of the vector $\be_2$ into $\mathbf{y}_1$ and $\mathbf{y}_2$, where $\mathbf{y}_i = (\mathbf{x}_1^{(i)} ,\mathbf{x}_2^{(i)} )$ for $i = 1, 2$. The gray areas denote the support of the vectors and the values inside the area are the corresponding Lee weights.}
\end{figure}
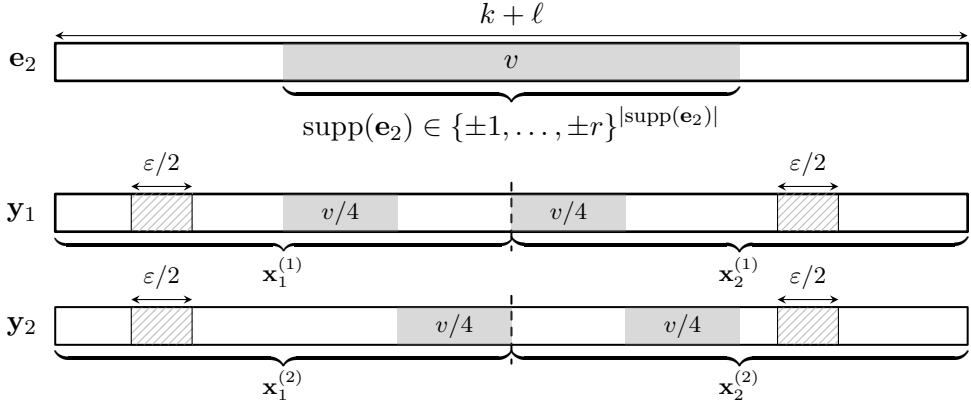

We now present the merging algorithms  and their asymptotic costs. 
For this, we fix the real numbers $V,L,E,U$ with $$0 \leq V \leq  \min\{T, \varphi(r,t,n,p^s)\}, \ \  0 \leq L \leq  1-R, \ \ 0 < E< R+L,$$ such that 
$0 \leq T-2V \leq M(1-R-L)$ and $0 <U < L.$
Then we fix the internal algorithm parameters and $v, \ell, \varepsilon,u $ which we see as functions depending on $n$, such that
$$ \lim\limits_{n \to \infty} \frac{v}{n}=V, \ \lim\limits_{n \to \infty} \frac{\ell}{n}=L, \ \lim\limits_{n \to \infty} \frac{\varepsilon}{n}=E \ \text{ and } \ \lim\limits_{n \to \infty} \frac{u}{n}=U .$$

\begin{algorithm}
\caption{Merge-concatenate}\label{algo:merge-concat}
\begin{flushleft}
Input: The input lists $\mathcal{B}_1, \mathcal{B}_2$, the positive integers $0 \leq u \leq \ell$,  $\bB_1, \bB_2 \in \zpsk{\ell \times (k+\ell)/2}$ and  $\bt \in \zpsk{\ell}$. \\ 
Output: $\mathcal{L} = \mathcal{B}_1 \concat_{\bt} \mathcal{B}_2$.
\end{flushleft}
\begin{algorithmic}[1]
\State Lexicographically sort $\mathcal{B}_1$  according to the last $u$ positions of $ \bx_1\bB_1^\top$ for $\bx_1 \in \mathcal{B}_1$. We also store the last $u$ positions of $ \bx_1\bB_1^\top $ in the sorted list.
\For{$\bx_2 \in \mathcal{B}_2$}
\For{$\bx_1 \in \mathcal{B}_1$ with $ \bx_1 \bB_1^\top =_u \bt - \bx_2\bB_2^\top$}
    \State $\mathcal L = \mathcal L \cup \{(\bx_1,\bx_2)\}$.
\EndFor
\EndFor
\State Return $\mathcal{L}.$
\end{algorithmic}
\end{algorithm}

\begin{lemma}[{\cite[Lemma 4.3]{leenp}}] \label{lemma:cost-mc}
The asymptotic of the average cost of Algorithm \ref{algo:merge-concat} is
\begin{align*} 
\lim\limits_{n \to \infty} \frac{1}{n} \max \left\{\log_{p^s}  \left(\mid \mathcal{B}_1 \mid\right), \log_{p^s}\left( \mid \mathcal{B}_2\mid \right),  \log_{p^s}  \left(\mid \mathcal{B}_1 \mid\right) +  \log_{p^s}  \left(\mid \mathcal{B}_2 \mid\right) -U\right\}.
\end{align*}
\end{lemma}

From this we get the lists
 \begin{align*}
    \mathcal{L}_1 = \mathcal{B}_1 \concat_0 \mathcal{B}_2, \
    \mathcal{L}_2 = \mathcal{B}_1 \concat_{\bs_2} \mathcal{B}_2.
\end{align*}

The second merge should not only merge to the target vector $\bs_2$, it should also check the Lee weight of the merged vector $\by_1+\by_2$ and also the Lee weight of the remaining error vector $\be_1 = \bs_1 -(\by_1+\by_2)\bA^\top$. 

\begin{algorithm}
\caption{Last Merge}\label{algo:last-merge}
\begin{flushleft}
Input: The input lists $\mathcal{L}_1, \mathcal{L}_2$, the positive integers $0\leq v \leq t,0 \leq u \leq \ell$,  $\bB \in \zpsk{\ell \times (k+\ell)}, \bs_2 \in \zpsk{\ell}$ and $\bs_1 \in \zpsk{n-k-\ell}, \bA \in \zpsk{(n-k-\ell) \times (k+\ell)}$. \\ 
Output: $\be \in \mathcal{L}_1 \bowtie \mathcal{L}_2$.
\end{flushleft}
\begin{algorithmic}[1]
\State Lexicographically sort $\mathcal L_1$ according to $ \by_1\bB^\top$ for $\by_1 \in \mathcal L_1$. We also store  $\by_1\bB^\top$ in the sorted list. 
\For{$\by_2 \in \mathcal{L}_2$}
\For{$\by_1 \in \mathcal{L}_1$ with $ \by_1\bB^\top = \bs_2 - \by_2\bB^\top$ }
    \If{$\lweight{\by_1+\by_2} = v$ and $\lweight{\bs_1-(\by_1+\by_2)\bA^\top}=t-v$}
    \State Return $(\bs_1-(\by_1+\by_2)\bA^\top, \by_1+ \by_2)$.
    \EndIf
\EndFor
\EndFor
\end{algorithmic}
\end{algorithm}

\begin{corollary}[{\cite[Corollary 2]{leenp}}]  \label{last-merge-cost}
The asymptotic average cost of the last  merge (Algorithm \ref{algo:last-merge}) is given by
\[\lim\limits_{n \to \infty} \frac{1}{n} \max \left \{ \log_{p^s}(\mid \mathcal{L}_1 \mid), \log_{p^s}(\mid \mathcal{L}_2 \mid ),    \left( \log_{p^s}(\mid \mathcal{L}_1 \mid) + \log_{p^s}(\mid \mathcal{L}_2 \mid )\right)-(L-U) \right\}.\] 
\end{corollary}
Note that the $L-U$ comes from the fact that the vectors already merge to $\bs_2$ on $U$ positions due to the first merge.
Also, it might happen that $\by_1 +\by_2$ results in a vector of Lee weight $v$, but the $\mathcal{E}$ positions did not cancel out, or the positions of low Lee weight are going above the threshold $r$. This will not be a problem for us, as this only results in a larger final list, which does not need to be stored and the success probability of the algorithm would then even be larger as $P.$

The way we choose $u$, is such that we ensure that there exists at least one representative $\by_1 \in \mathcal{L}_1$ of the solution $\be_2$, i.e., such that there exists $\by_2 \in \mathcal{L}_2$ with $\by_1 + \by_2 =\be_2.$
Thus, we have to compute the expected total number of such representatives for a fixed $\be_2.$ From Lemma \ref{lemma:e2Distribution}, we know that $\be_2$ follows a uniform distribution in $S_{(r)}(0,v,k+l,p^s)$.

Using the marginal distribution in \eqref{equ:marginal} and \eqref{equ:prob_LW}, we can compute the expected Lee weight distribution for $\be_2$. Let $\lambda$ be the expected Lee weight composition of $\be_2$, and $\sigma$ be the expected support size of $\be_2$. 
Also recall that for a weak composition $\lambda$ of $v$, we denote by $C(v/2, v, \lambda, k+\ell ,p^s)$ the number of weak ompositions $\pi$ of $v/2$ which fit into a composition $\lambda$ of length $k+\ell$, i.e., the maximal part sizes are given by $\lambda.$
\begin{lemma}
    The expected number of representatives $(\by_1,\by_2) \in \mathcal{L}_1 \times \mathcal{L}_2$ for a fixed solution $\be_2$ is at least given by 
       \[C(v/2, v, \lambda, k+\ell,p^s) \binom{k+\ell - \sigma}{\varepsilon}(p^s-1)^\varepsilon,\] where $\lambda$ is the expected Lee weight composition of $\be_2$, and $\sigma$ is the expected support size of $\be_2$.
\end{lemma}
\begin{proof}
Consider the Lee weight composition of $\be_2$ to be $\lambda=(\lambda_1, \ldots, \lambda_{k+\ell}),$ which is such that $\lambda_i=\text{wt}_L((\be_2)_i)$.
Thus, $\be_2= ( s_1 \lambda_1, \ldots, s_{k+\ell} \lambda_{k+\ell}),$ for $s_i \in \{1,-1\}.$
Then, to get all possible representatives $\by_1$, we need the number of weak compositions $\pi$ of  $v/2$ fitting into $\lambda.$ In fact, for any $\pi = (\pi_1, \ldots, \pi_{k+\ell})$ fitting into $\lambda$, there will exist exactly one eligible $\by_1$ with $\text{wt}_L((\by_1)_i)=\pi_i$ and $(\by_1)_i = s_i \pi_i.$
Note that the Lee weight composition of $\by_2 \in \mathcal{L}_2$ is then $$\mid \lambda-\pi \mid = (\mid \lambda_1 - \pi_1 \mid, \ldots, \mid \lambda_{k+\ell} - \pi_{k+\ell} \mid).$$
   
On the other hand, for any representative $\by_1$, we cannot have $\pi_i=\text{wt}_L((\by_1)_i) > \text{wt}_L((\be_2)_i)$ and $(\by_1)_i = -s_i \pi_i$ for any $i \in \{ 1,\ldots,\sigma\}$.
In fact, let us assume we have $A$ many positions in $\by_1$ which are such that $\pi_i=\text{wt}_L((\by_1)_i) > \text{wt}_L((\be_2)_i)= \lambda_i$. Then due to the entry-wise additivity of the Lee weight, we have that  $\by_2$, with composition $\mid \lambda-\pi \mid$, has $\text{wt}_L(\by_2) > v/2:$ in the  considered $A$ positions we have that $\text{wt}_L((\by_2)_j)= \pi_j-\lambda_j$ and the Lee weight of the remaining $\sigma-A$ positions is given by $\text{wt}_L((\by_2)_j) =  \lambda_j-\pi_j$, which if we sum over all positions gives 
\begin{align*} 
   \text{wt}_L(\by_2) & =\sum_{j=1}^A ( \pi_j-\lambda_j ) + \sum_{j=A+1}^{k+\ell} (\lambda_j-\pi_j) \\
   &= \sum_{j=1}^A ( \pi_j-\lambda_j )+ v- \sum_{j=1}^A \lambda_j  - \left(v/2 -  \sum_{j=1}^A - \pi_j \right) \\
   & = v/2 + 2 \left( \sum_{j=1}^A  \pi_j-\lambda_j \right)  \neq v/2. \end{align*}
It is easy to see, that for each fixed $\pi$, there exists only one representative $\by_1$, which has in each position the same sign as $\be_2.$

Recall that $C(v/2, v, \lambda, k+\ell,p^s)$  denotes  the number of weak compositions $\pi$ of $v/2$ which fit into $\lambda$.  
Now, since $\by_1$ can take any non-zero value on the $\varepsilon$ positions outside of the support of $\be_2$, we get the claim. Finally, the exact number of representations might even be larger than this, since a solution $\be_2$ might also be formed from positions $\mathcal{E}$ which will not cancel out, as assumed for this computation.  
\end{proof}
 
In order to ensure the existence of at least one representative $\by_1 \in \mathcal{L}_1$ of $\be_2$, we now choose 
\[ u=\left\lfloor\log_{p^s}\left(C(v/2, v, \lambda, k+\ell,p^s) \binom{k+\ell-\sigma}{\varepsilon} (p^s-1)^\varepsilon\right)\right\rfloor.\]
Thus, in the asymptotic cost we need to compute 
$U=\lim\limits_{n \to \infty} u/n.$ 

Let us denote the asymptotics of the binomial coefficient by 
\begin{align*} H(F, G) & := \lim_{n\to \infty} \frac{1}{n} \log_{p^s} \left(  \binom{f(n)}{g(n)} \right) \\ 
& = F \log_{p^s}(F) - G \log_{p^s}(G) - (F-G) \log_{p^s} (F-G), \end{align*} 
where $f(n), g(n)$ are integer-valued functions such that $\lim\limits_{n\to \infty} \frac{f(n)}{n} = F$ and $\lim\limits_{n \to \infty} \frac{g(n)}{n} = G$. By Lemma \ref{asympt_smallballs}, we have computed \[ \gamma(v/2) = \lim\limits_{n' \to \infty} \frac{1}{n'}\log_{p^s}\left(C(v/2, v, \lambda, n',p^s)\right). \] For us $n'=k+\ell$, which also tends to infinity for $n$ going to infinity. Thus, \[\lim\limits_{n \to \infty} \frac{k+\ell}{n} \lim\limits_{k +\ell \to \infty} \frac{1}{k+\ell} \log_{p^s}\left(C(v/2, v, \lambda, k+\ell,p^s)\right) = (R+L) \gamma(v/2). \]
 Then, 
\begin{align*} U = (R+L) \gamma(v/2) + H(R+L-S, E) +E,\end{align*}
where $S=\lim\limits_{n \to \infty} \sigma/n.$

\begin{algorithm}[h!]
\caption{Lee-BJMM with Small Balls}\label{algo:bjmm}
\begin{flushleft}
Input: $\bH \in \left(\mathbb{Z}/p^s\mathbb{Z}\right)^{(n-k)\times n}, \bs \in \left(\mathbb{Z}/p^s\mathbb{Z}\right)^{n-k}, t \in \mathbb{N}$, given the positive integers $0 \leq \ell \leq n-k, 0 \leq v \leq t, 0 \leq \varepsilon \leq k+\ell$. \\ 
Output: $\be \in \left(\mathbb{Z}/p^s\mathbb{Z}\right)^n$ with $\bs= \be\bH^\top$ and $\text{wt}_L(\be)=t$.
\end{flushleft}
\begin{algorithmic}[1]
\State Choose a $n \times n$ permutation matrix $\bP$ and find an invertible matrix $\bU \in \left(\mathbb{Z}/p^s\mathbb{Z}\right)^{(n-k) \times (n-k)}$ such that 
\[\bU\bH\bP= \begin{pmatrix} \text{Id}_{n-k-\ell} & \bA \\ 0 & \bB \end{pmatrix}, \] where $\bA  \in \left(\mathbb{Z}/p^s\mathbb{Z}\right)^{(n-k-\ell) \times (k+\ell)}, \bB \in \left( \mathbb{Z}/p^s\mathbb{Z}\right)^{\ell \times (k+\ell)}.$
\State Compute \[\bs\bU^\top = \begin{pmatrix} \bs_1 & \bs_2 \end{pmatrix}, \] where $  \bs_1 \in \left( \mathbb{Z}/p^s\mathbb{Z}\right)^{n-k-\ell}, \bs_2 \in \left(\mathbb{Z}/p^s\mathbb{Z}\right)^\ell.$ 
\State Choose a set $\mathcal{E} \subset \{1, \ldots, k+\ell\}$ of size $\varepsilon.$
\State Build the lists $\mathcal{B}_1, \mathcal{B}_2$ as 
\[\mathcal{B}_i=\{\bx \mid \bx_{\mathcal{E}^C} \in \{0, \pm1, \ldots, \pm r\}^{(k+\ell-\varepsilon)/2}, \bx_\mathcal{E} \in  \left(\mathbb{Z}/p^s\mathbb{Z}\right)^{\varepsilon/2}, \text{wt}_L(\bx_{\mathcal{E}^C})=v/4\}. \]
\State Compute  $   \mathcal{L}_1 = \mathcal{B}_1 \concat_0 \mathcal{B}_2$ and $ 
    \mathcal{L}_2 = \mathcal{B}_1 \concat_{\bs_2} \mathcal{B}_2.$
    \State Compute  $\be \in \mathcal{L}_1 \bowtie \mathcal{L}_2$.
    \State If this fails, return to Step 1.
    \State Return $\bP^\top\be$.
\end{algorithmic}
\end{algorithm}

To ease the notation, we will denote the asymptotics of the restricted Lee-metric sphere by
 \[A_{(r)}(t,p^s) \coloneqq \lim\limits_{n \to \infty}  \frac{1}{n} \log_{p^s}\left( F_{(r)}(t(n), n, p^s) \right).\]
Further, let us denote by $W= \psi(r,t,n,p^s)/n.$

     \begin{theorem}\label{thm:cost}
     The asymptotic average time complexity of the Lee-metric BJMM algorithm on two levels is at most given by
     $I + C,$ where 
     \[I = A_{(M)}(t,n,p^s)- A_{(r)}(v, k+\ell,p^s) - A_{(M)}(t-v,n-k-\ell,p^s),\] is the expected number of iterations and 
$        C=  \max\{  B, 2D -L + U, D\} $
    is the expected cost of one iteration with 
    \begin{align*}
       B &= A_{(r)}(v/4,(k+\ell-\varepsilon)/2, p^s) +H((R+L)/2,E/2)  + E/2, \\
      D &= A_{(r)}(v/2, k+\ell-\varepsilon,p^s) + H(R+L,E)+ E-U.
    \end{align*}
      In addition, we have an expected memory of at most
$ \mathcal{M} = \max\{ B, D  \}. $
    On a capable quantum computer, the average time complexity is given by at most
    $$I/2 + \max\left\{B,D,\frac{1}{2}(2D-L+U)\right\}.$$
    \end{theorem}

    \begin{proof}

For our base lists $\mathcal{B}_i$, we have that 
  \[ \mid \mathcal{B}_i \mid = F_{(r)}(v/4, (k+\ell-\varepsilon)/2,  p^s) \binom{(k+\ell)/2}{\varepsilon/2} (p^s-1)^{\varepsilon/2}.\]
Due to Lemma \ref{asympt_V}, the cost of the first merge is then given by
\begin{align*} B &= A_{(r)}(v/4, (k+\ell-\varepsilon)/2,  p^s) + H((R+L)/2,E/2)+ E/2.
\end{align*} 

    For the second merge we also need to compute the asymptotic sizes of $\mathcal{L}_i.$
    First, we note that 
    \[\mid \mathcal{L}_i \mid= \frac{F_{(r)}(v/2, k+\ell-\varepsilon, p^s) \binom{k+\ell}{\varepsilon} (p^s-1)^\varepsilon}{p^{su}}.\]
    Thus,
    \begin{align*} D = & \lim\limits_{n \to \infty} \frac{1}{n} \log_{p^s} \left( \mid \mathcal{L}_i\mid\right) = A_{(r)}(v/2, k+\ell-\varepsilon,p^s) + H(R+L,E) + E-U.
    \end{align*}
    Using Corollary \ref{last-merge-cost}, the  second merge costs asymptotically
    \begin{align*}
        2D-L+U = & 2A_{(r)}(v/2,k+\ell-\varepsilon,p^s) +2H(R+L,E)+ 2E -U -L.
    \end{align*}
    
    We recall that the success probability of the algorithm is given by $P$, 
hence for $0 \leq V \leq \varphi(r,t,n,p^s)/n$, we get the following asymptotic number of iterations
    \begin{align*}
       A_{(M)}(t,n,p^s) - A_{(r)}(v, k+\ell,p^s) - A_{(M)}(t-v,n-k-\ell,p^s) .
    \end{align*}
    
  The average memory required for the algorithm is given by $\mid \mathcal{B}_i \mid$ and $\mid \mathcal{L}_i \mid,$ thus taking the asymptotic of these lists the claim follows.

    Finally, note that Grover's algorithm can be used to speed up on a capable quantum computer whenever a list $L$ has to be searched. In particular, instead of $\mathcal{O}(\mid L \mid)$, Grover's algorithm only requires $\mathcal{O}(\sqrt{\mid L \mid})$ operations. Thus, this results asymptotically in $\lim_{n \to \infty} \frac{1}{2n} \log_{p^s}(\mid L \mid).$ In our classical asymptotic cost, every term stems from a searched list, except for $B$ and $D$, which are intermediate lists that have to be stored in full. 
     
\end{proof}
  Observe that $\ell, v,r, \varepsilon$ are  internal parameters, which can be chosen optimal, i.e., such that the algorithm achieves the minimal cost. Clearly, the choice for the threshold $r$ will influence the possible choices for $v.$\\

  \noindent\textit{The Amortized Case:} \ If we only consider $p^{su}$ many vectors from the base lists $\mathcal{B}_i$, we could potentially reduce the cost and memory.
  
  The algorithm is going to work exactly the same way, with the only difference that the base lists $\mathcal{B}_i'$ have size $p^{su}$. Thus, after using the merging Algorithm \ref{algo:merge-concat} on $u$ positions we get lists $\mathcal{L}_i'$ of size $p^{su}$ as well. Finally, we merge these lists using Algorithm \ref{algo:last-merge} on $\ell$ positions.
  Note that the conditions on $U= \lim_{n \to \infty} u(n)/n$ are 
\[     L/3 \leq  U \leq \min\{(R+L) \gamma(v/2) + H(R+L-S,E)+E, B,L\},\] where $B$ denotes the asymptotic size of the original base lists, i.e.,
  \[ B = A_{(r)}(v/4,(k+\ell-\varepsilon)/2, p^s) +H((R+L)/2, E/2) + E/2. \]
The condition $L/3 \leq U$, comes from the size of the final list, i.e., the number of solutions for the smaller instance, which is $\frac{p^{s2u}}{p^{s(\ell-u)}} = p^{s(3u-\ell)}.$ In order to have at least one solution, we require $3u \geq \ell.$
Recall that  $(R+L) \gamma(v/2) + H(R+L-S,E)+E$ denotes the asymptotic number of representations, thus the condition $U \leq (R+L) \gamma(v/2) + H(R+L-S,E)+E$ is the same as for the original algorithm. The condition $U \leq B$, as well as $U \leq L$ are straightforward. 
  
  Note that in the amortized case, the success probability of splitting $\be=(\be_1,\be_2)$ is not simply given by  \[P=F_{(r)}(v,k+\ell,p^s)F(t-v,n-k-\ell,p^s)F(t,n,p^s)^{-1}\] as in the non-amortized case, since our list of $\be_2$ is by construction smaller. That is instead of all solutions to the smaller problem $F_{(r)}(v,k+\ell,p^s)p^{-s\ell}$, we only consider $Z$ many solutions to the smaller problem.  In other words, $Z$ is  the number of distinct $\be_2$ in our last list.
Similar to the approach of \cite{thomas}, we have a success probability of 
  \[P'= Zp^{s\ell}  F(t-v,n-k-\ell,p^s)F(t,n,p^s)^{-1}.\]
 In order to compute $Z$, let us denote by $X$ the maximal amount of collisions of the last merge which would lead to an $\be_2$ (that is with possible repetitions), by $Y$ the total number of solutions to $\be_2\bB^\top=\bs_2$ with $\be_2 \in S_{(r)}(v,k+\ell,p^s)$, namely 
 \[Y= F_{(r)}(v,k+\ell,p^s) p^{-s\ell},\] and finally by $W$ the number of collisions that we are considering, that is 
 \[W= p^{s(3u-\ell)} = p^{s2u} p^{-s(\ell-u)}.\]
 This leaves us with a combinatorial problem: having a basket with $X$ balls having $Y$ colors, if we pick $W$ balls at random, how many colors are we going to see on average? This will determine the number of distinct tuples $\be_2$ in the final list. This number is on average 
 \[Y\left( 1- \binom{X-X/Y}{W} \binom{X}{W}^{-1}\right),\]
 which can be lower bounded by $W.$
 In fact, 
 \begin{align*}
     1-\binom{X-X/Y}{W}\binom{X}{W}^{-1}     = 1- \frac{(X-X/Y+1-W) \cdots (X-W)}{(X-X/Y+1) \cdots X} \\  \geq 1-(1-W/X)^{X/Y} \sim W/Y.
 \end{align*}
  Hence, $Z \geq p^{s(3u-\ell)}$ and we get a success probability of at least 
  \[p^{s3u}  F(t-v,n-k-\ell,p^s)F(t,n,p^s)^{-1}.\]
  The asymptotic cost of the amortized version of Algorithm \ref{algo:bjmm} is then given by 
  $I'+ \max\{U,3U-L\},$ where $I'$ is   the expected number of iterations, i.e., \[I' \leq A_{(M)}(t,n,p^s)- 3U - A_{(M)}(t-v,n-k-\ell,p^s).\]
  Hence, we can see that the restriction to the smaller balls does not influence the amortized version of BJMM, as the idea of amortizing is already to restrict the balls. The restriction only influences the conditions and thus the possible choices of $U.$
  
  \subsection{Decoding beyond the Minimum Distance}\label{bigball}

There could be scenarios where one wants to decode more errors than the minimum Lee distance of the code at hand allows. In the classical case, i.e., in the Hamming metric,  the cost can then be divided by the expected number of solutions $N$. This follows from the fact that for each of the $N$ solutions we have a success probability $P$ for one iteration to succeed. Assuming that the solutions are independent, this implies that to find one solution we expect the number of iterations to be $\frac{1}{PN}$. 
  
  In a scenario where we have $t>Mn/2$, the marginal distribution of $\be \in \left(\mathbb{Z}/p^s\mathbb{Z}\right)^n$ implies that  $ \pm M$ is the most likely entry of  $\be$, then the second most likely is $\pm (M-1)$ and so on, until the least likely entry is $0.$ 
In this case, we will reverse the previous algorithm and for some threshold Lee weight $0\leq r \leq M$, we want the vector $\be_2$ of Lee weight $t-\varphi(r-1,t,n,p^s) \leq v \leq t$ to live in $\{\pm r, \ldots, \pm M\}^{k+\ell}$. In order to construct such a vector, we will use a similar construction as before, where we exchange the set $\{ 0, \pm 1, \ldots, \pm r\}$ with $\{\pm r, \ldots, \pm M\}$. Note that  the success probability of such splitting is now given by 
\[P = F^{(r)}(v,k+\ell,p^s) F(t-v,n-k-\ell,p^s) F(t,n,p^s)^{-1}.\]

Let us first illustrate the idea and then compute the sizes of the lists involved.
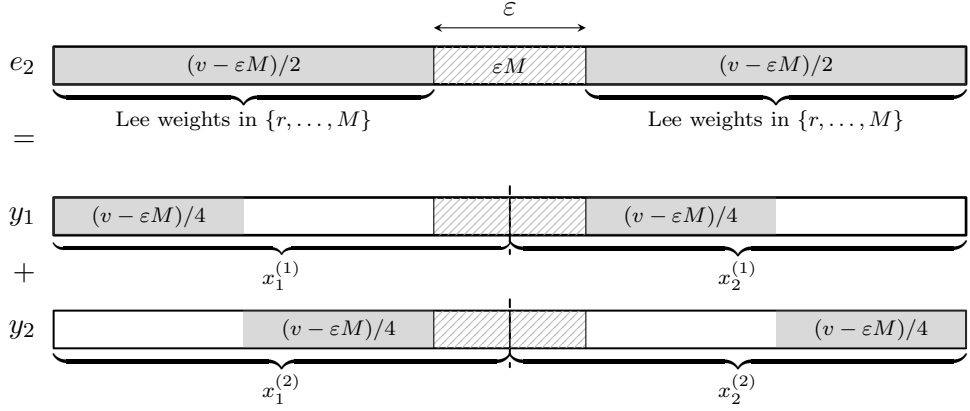
\begin{figure}[h!]
    \begin{center}
        \begin{tikzpicture}[line cap=round,line join=round,>=triangle 45,x=1cm,y=1cm]
            \draw (-1.1, 3.75) node[left] {$e_2 $};
            \draw[line width=1pt] (-1,4) -- (-1, 3.5) -- (11,3.5) -- (11,4) -- cycle;
            
            \draw[stealth-stealth] (4, 4.25) -- (6, 4.25);
            \draw (5, 4.5) node[] {$\varepsilon$};
            
            \fill[line width=1pt,fill = gray, fill opacity=0.3] (-1,3.5) rectangle (4,4);
            \draw (1.5,3.75) node[] {\scriptsize $(v- \varepsilon M)/2$};
            \draw (1.5, 3.6) node[below] {$\underbrace{\hspace{5cm}}$};
            \draw (1.5, 3.3) node[below] {\scriptsize Lee weights in $\set{r, \dots, M}$};
            
            \fill[line width=1pt,fill = gray, fill opacity=0.3] (11,3.5) rectangle (6,4);
            \draw (8.5, 3.75) node[] {\scriptsize $(v- \varepsilon M)/2$};
            \draw (8.5, 3.6) node[below] {$\underbrace{\hspace{5cm}}$};
            \draw (8.5, 3.3) node[below] {\scriptsize Lee weights in $\set{r, \dots, M}$};
            
            \draw[pattern = north east lines, pattern color = lightgray] (4, 4) rectangle (6, 3.5);
            \draw (5,3.75) node[] {\scriptsize $\varepsilon M$};
        
        \draw (-1.1, 2.75) node[left] {$=$};
        \draw (-1.1, 1) node[left] {$+$};
            \draw (-1.1,1.75) node[left] {$y_1$};
            \draw[line width=1pt] (-1,1.5) -- (-1,2) -- (11,2) -- (11,1.5) -- cycle;
            \draw[line width = 0.75pt, style = dashed] (5, 2.15) -- (5, 1.25);
            
            \fill[line width=1pt,fill = gray, fill opacity=0.3] (-1,1.5) rectangle (1.5,2);
            \draw (0.25,1.75) node[] {\scriptsize $(v- \varepsilon M)/4$};
            \draw (2, 1.6) node[below] {$\scriptsize\underbrace{\hspace{6cm}}$};
            \draw (2, 1.3) node[below] {\scriptsize $x_1^{(1)}$};
            \draw[pattern = north east lines, pattern color = lightgray] (4,1.5) rectangle (5,2);

            \fill[line width=1pt,fill = gray, fill opacity=0.3] (6,1.5) rectangle (8.5,2);
            \draw (7.25, 1.75) node[] {\scriptsize $(v- \varepsilon M)/4$};
            \draw (8, 1.6) node[below] {$\underbrace{\hspace{6cm}}$};
            \draw (8, 1.3) node[below] {\scriptsize $x_2^{(1)}$};
            \draw[pattern = north east lines, pattern color = lightgray] (5,1.5) rectangle (6,2);

            \draw (-1.1,0.25) node[left] {$y_2 $}; 
            \draw[line width=0.75pt] (-1,0.5) -- (-1,0) -- (11,0) -- (11,0.5) -- cycle;
            \draw[line width = 0.75pt, style = dashed] (5, 0.65) -- (5, -0.25);
    
            \fill[line width=1pt,fill = gray, fill opacity=0.3] (4,0) rectangle (1.5,0.5);
            \draw (2.75, 0.25) node[] {\scriptsize $(v- \varepsilon M)/4$};
            \draw (2, 0.1) node[below] {$\underbrace{\hspace{6cm}}$};
            \draw (2, -0.2) node[below] {\scriptsize $x_1^{(2)}$};
            \draw[pattern = north east lines, pattern color = lightgray] (5,0.5) rectangle (4,0);
    
            \fill[line width=1pt,fill = gray, fill opacity=0.3] (11,0) rectangle (8.5,0.5);
            \draw (9.75, 0.25) node[] {\scriptsize $(v- \varepsilon M)/4$};
            \draw (8, 0.1) node[below] {$\underbrace{\hspace{6cm}}$};
            \draw (8, -0.2) node[below] {\scriptsize $x_2^{(2)}$};
            \draw[pattern = north east lines, pattern color = lightgray] (6,0.5) rectangle (5,0);
        \end{tikzpicture}
    \end{center}
\caption{Illustration of two levels decomposition of the vector $\be_2$ into $\mathbf{y}_1$ and $\mathbf{y}_2$, where $\mathbf{y}_i = (\mathbf{x}_1^{(i)} ,\mathbf{x}_2^{(i)} )$ for $i = 1, 2$. The gray areas denote the support of the vectors and the values inside the area are the corresponding Lee weights. For $(\mathbf{y}_1)_i$ and $(\mathbf{y}_2)_i$ with $i \in \mathcal{E}$, we require $\lweight{(\mathbf{y}_1)_i + (\mathbf{y}_2)_i} = M$.}
\end{figure}
Note, that one of the main differences to the previous algorithm is that we require to partition the weights in order to guarantee that the large weight entries of $\by_1$ will not be decreased after adding $\by_2$. For this let us introduce the following set of indices $Z_1,Z_2, W_1,W_2,\mathcal{E}_1,\mathcal{E}_2$ satisfying
$\mid Z_i \mid = \mid W_i \mid = (k+\ell-\varepsilon)/4,\, Z_1 \cap Z_2 = \emptyset, W_1 \cap W_2 = \emptyset, \mathcal{E}_1\cap \mathcal{E}_2=\emptyset \text{ and } Z_i \cap W_i \cap \mathcal{E}_i = \emptyset \text{ for } i \in \set{1, 2}.$
Let us denote their union by $\mathcal{E}= \mathcal{E}_1 \cup \mathcal{E}_2, Z=Z_1\cup Z_2, W=W_2\cup W_2.$  For $i \in \set{1, 2}$, the base lists $\mathcal{B}_i$ are then
\begin{align*}
    \mathcal{B}_i =  \Big\{\nu_i(\bx) \ \big|  & \ \bx_{Z_i}  \in  \{0\}^{(k+\ell-\varepsilon)/4}, \bx_{W_i}  \in  \{\pm r, \ldots, \pm M\}^{(k+\ell-\varepsilon)/4},  \bx_{\mathcal{E}_i}  \in  \left(\mathbb{Z}/p^s\mathbb{Z}\right)^{\varepsilon/2}, \\ & \text{wt}_L(\bx_{W_i})= (v-\varepsilon M)/4, \nu_i  \in  S_{(k+\ell)/2} \Big\}.
\end{align*} 
 All of the base lists have the same size, which is given by 
 \[\binom{ (k+\ell)/2}{\varepsilon/2} p^{s\varepsilon/2}\binom{(k+\ell-\varepsilon)/2}{(k+\ell-\varepsilon)/4} F^{(r)}\big((v-\varepsilon M)/4, (k+\ell-\varepsilon)/4,p^s\big).\]
Performing the concatenation merge of Algorithm \ref{algo:merge-concat}, we build $\mathcal{L}_1$ and $\mathcal{L}_2$ from $\mathcal{B}_1$ and $\mathcal{B}_2$ as
 \begin{align*}
     \mathcal{L}_1  =  \Big\{ \mu_1(\by_1)  \ \big| \  & \mu_1  \in  S_{k+\ell},\by_1\bB^\top =_u \mathbf{0}, (\by_1)_{Z}  \in  \{0\}^{(k+\ell-\varepsilon)/2},(\by_1)_{\mathcal{E}}  \in  \left( \mathbb{Z}/p^s\mathbb{Z} \right)^{\varepsilon},\\ &   (\by_1)_{W}  \in  \{\pm r, \ldots, \pm M\}^{(k+\ell-\varepsilon)/2}, \text{wt}_L((\by_1)_{W})= (v-\varepsilon M)/2 \Big\}, \\
      \mathcal{L}_2  =  \Big\{ \mu_2(\by_2) \ \big| \ & \mu_2  \in  S_{k+\ell}, \by_2\bB^\top =_u \mathbf{s}_2,  (\by_2)_{Z}  \in  \{0\}^{(k+\ell-\varepsilon)/2},(\by_2)_{\mathcal{E}}  \in  \left( \mathbb{Z}/p^s\mathbb{Z} \right)^{\varepsilon}, \\ & (\by_2)_{W}  \in  \{\pm r, \ldots, \pm M\}^{(k+\ell-\varepsilon)/2},   \text{wt}_L((\by_2)_{W})= (v-\varepsilon M)/2  \Big\}.
 \end{align*} 
 Both lists are of size
 \[ \binom{k+\ell}{\varepsilon}p^{s(\varepsilon-u)}\binom{k+\ell-\varepsilon}{(k+\ell-\varepsilon)/2}F^{(r)}\big((v-\varepsilon M)/2, (k+\ell-\varepsilon)/2,p^s\big).\]
 For this procedure to work, we also need the additional condition on $v,r$ and $\varepsilon$, that
 \[v \geq \varepsilon(M-r)+ r(k+\ell).\]
 Then, a final merge using Algorithm \ref{algo:last-merge} will produce a final list of all smaller solutions of the smaller instance which does not require to be stored. 

\begin{lemma}
  The number of representations $\be_2=\by_1+\by_2$ for $(\by_1, \by_2) \in \mathcal{L}_1 \times \mathcal{L}_2$ is then given by at least
 \begin{align*}
     R_B=  & \binom{k+\ell}{\varepsilon} \left( \sum_{i=0}^\varepsilon \binom{\varepsilon}{i} (M-r+1)^i r^{\varepsilon-i} \right.  \left. \binom{\varepsilon'}{i} (M-r+1)^i \binom{\varepsilon'-i}{(\varepsilon'-i)/2} \right),
 \end{align*}
 for $\varepsilon'=k+\ell-\varepsilon.$
\end{lemma}
\begin{proof}
To give a lower bound on the number of representations it is enough to give one construction.\\
The overall idea of this construction is to split the $\mathcal{E}_1$ positions of $\by_1$ and $\mathcal{E}_2$ positions of $\by_2$ into those parts where they overlap and those parts where they do not overlap. In the parts where $\mathcal{E}_1$ does not overlap with $\mathcal{E}_2$, we can only allow small Lee weights in $\by_1$ such that, by adding large Lee weight entries of $\by_2$, we can still reach the large Lee weight entries of $\be_2.$ 

So let us consider a fixed $\be_2 \in F^{(r)}(v, k+\ell,p^s).$ As a first step we fix the $\mathcal{E}_1$ positions which gives $\binom{k+\ell}{\varepsilon}$. Then, within the $\mathcal{E}_1$ position we fix those of small Lee weight. This means for a fixed position we can assume that the entry in $\be_2$ is $a$ with $r \leq \text{wt}_L(a) \leq M$. Small Lee weights of $\by_1$ now refer to the possible values of $\by_1$ in this position such that $a$ can be reached through large Lee weight entries of $\by_2$. That is, for example if $a=r$, we allow in $\by_1$ the entries $\{0, -1, \ldots, r-M\}$, or if $a=M$ we allow in $\by_1$ the entries $\{M-r, \ldots, 0\}.$ These allowed sets of small Lee weight always have size $M-r+1$, independently of the the value $a$. Thus, in $\mathcal{E}_1$  of size $\varepsilon$ we choose $i$ entries of small Lee weight, which give $\binom{\varepsilon}{i}(M-r+1)^i$ many choices. For the remaining $\varepsilon-i$ positions in $\mathcal{E}_1$ we have large Lee weights in $\by_1$, which cannot reach the large Lee weight entries of $\be_2$ through large Lee weight entries in $\by_2$. Thus, they must come for the $\mathcal{E}_2$ positions. In these entries we have $r^{\varepsilon-i}$ possible choices. Note that out of the $\varepsilon$ many positions of $\mathcal{E}_2$  we have only assigned $\varepsilon-i$ many. Hence, as a next step we choose of the remaining $k+\ell-\varepsilon$ positions the remaining $i$ positions to have small Lee weight in $\by_2$. Thus, the fixed large Lee weight entries of $\be_2$ can be reached by adding these positions to large Lee weight entries of $\by_1$. For this we have $\binom{k+\ell-\varepsilon-i}{i}(M-r+1)^i$ possibilities. As a final step we then partition the remaining positions to either be 0 or of large Lee weight, i.e., $\binom{k+\ell-\varepsilon-i}{(k+\ell-\varepsilon-i)/2}$.  
 \end{proof}
 Thus, we will need the additional condition $\varepsilon \leq (k+\ell)/2$ and we  choose \[ u = \left\lfloor \log_{p^s}(R_B) \right\rfloor.\]
Since we cannot take the asymptotic of an infinite sum, we need to bound this  quantity. In fact, setting $i=\varepsilon$ gives such lower bound.
\[R_B \geq \binom{k+\ell}{\varepsilon} (M-r+1)^{2\varepsilon}   \binom{k+\ell-\varepsilon}{\varepsilon} \binom{k+\ell-2\varepsilon}{(k+\ell-2\varepsilon)/2}.\]
 Then, \begin{align*} U   = \lim_{n \to \infty} U(n)/n = &  H(R+L,E)+2E\log_{p^s}(M-r+1) \\ & +H(R+L-E,E)+H(R+L-2E, (R+
 L-2E)/2). 
 \end{align*}
 
In addition, since we decode beyond the minimum distance, the LSDP has several solutions. Since the inputs have been chosen uniform at random, we can assume that these solutions are independent from each other. 
Thus, to find just one of all the expected 
\[N = \frac{F(t,n,p^s)}{p^{s(n-k)}}\] solutions we have an expected number of iterations given by $(NP)^{-1},$ instead of $P^{-1}$. Note that asymptotically this value is bounded by $R$, as
\[ X= \lim_{n \to \infty} \frac{1}{n} \log_{p^s} \left( F(t,n,p^s)p^{-s(n-k)} \right) = A_{(M)}(t,n,p^s)-1+R \leq   R. \]
Let us denote by $A^{(r)}(t,n,p^s) = \lim\limits_{n \to \infty} 1/n \log_q( F^{(r)}(t,n,p^s))$.
     \begin{corollary}
     The asymptotic average time complexity of the Lee-metric BJMM algorithm on two levels for $t>Mn/2$ is given by at most $I+C,$ where
     \begin{align*} I= & (1-R) -A^{(r)}(v,k+\ell,p^s) -A_{(M)}(t-v,n-k-\ell,p^s)\end{align*}
     is the expected number of iterations and  $C= \max\{B, D, 2D-L+U\}$ is the cost of one iteration, where 
     \begin{align*}
         B =E/2+ &  H((R+L)/2,E/2) +H((R+L-E)/2,(R+L-E)/4) \\ & +A^{(r)}((v-\varepsilon M)/4, (k+\ell-\varepsilon)/4,p^s), \\
         D=E-U+ &  H(R+L,E) +H(R+L-E,(R+L-E)/2)\\ & +A^{(r)}((v-\varepsilon M)/2, (k+\ell-\varepsilon)/2,p^s).
     \end{align*}
      In addition, we have an expected memory of at most  $\mathcal{M}= \max\{B,D \}.$ 
    On a capable quantum computer, the average time complexity is given by at most
    $I/2 +\max\{B,D, \frac{1}{2}(2D-L+U) \}.$
    \end{corollary}

  \noindent \textit{The Amortized Case:} \ We  consider again the amortized version of this algorithm, i.e., we only take $p^{su}$ many vectors from the base lists $\mathcal{B}_i^{(1)}$, respectively $\mathcal{B}_i^{(2)}$.
  
  The algorithm is going to work exactly the same way,  similar to the amortized version for the first scenario. 
  The asymptotic cost of the amortized version of Algorithm \ref{algo:bjmm} is then given by 
  $I'+ \max\{U,3U-L\},$ where $I'$ is as before the expected number of iterations, i.e., \begin{align*} I' \leq & (1-R)-  3U - A_{(M)}(t-v,n-k-\ell,p^s). 
  \end{align*}

\section{Comparison}\label{sec:comparison}
In this section we want to see how much cost reduction we were able to achieve by using this additional information on the error vector. For this we will compare the new Lee-metric BJMM algorithm to the Lee-metric BJMM algorithm from \cite{leenp} and to the algorithm using Wagner's approach in \cite{thomas}, which were until now the fastest algorithms to solve the LSDP.
We denote by $e(R,p^s)$ the exponent of the asymptotic cost and compare $e(R^*,p^s)$ for $R^* = \arg\max\limits_{0 \leq R \leq 1} \left( e(R,p^s) \right)$.

In the first scenario, we only decode up to the Gilbert-Varshamov bound, i.e., we consider 
$V(d(n),n,p^s) = 1-R.$
Hence, we give an immediate relation between $T$ and $R,$ where $T$ is  $\lim_{n \to \infty} d(n)/n$, i.e., we are considering full-distance decoding. 

\begin{figure}
    \centering
    \includegraphics[width=\textwidth]{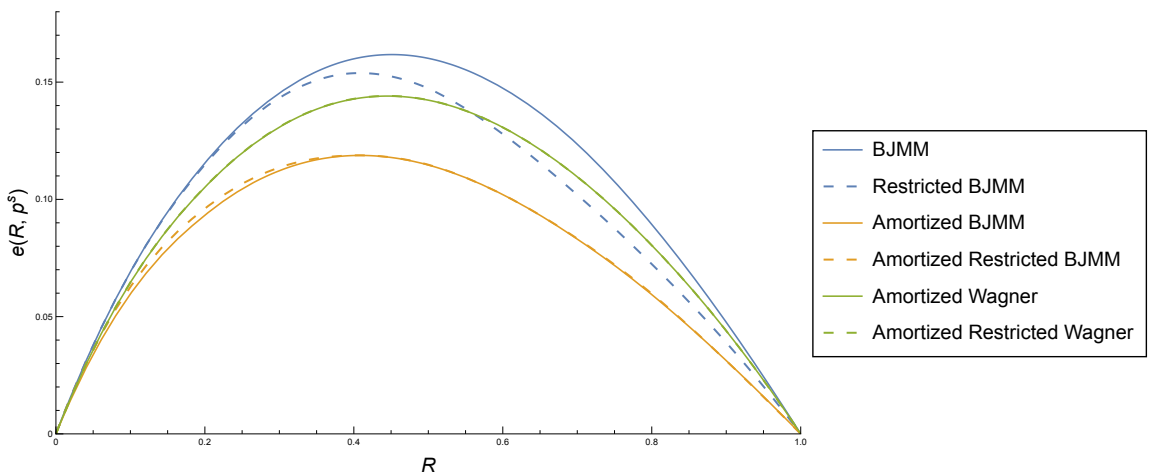}
    \caption{Comparison of asymptotic costs of full-distance decoding of different algorithms and their restricted versions, for $p=47, s=1$ and assuming the asymptotic Gilbert-Varshamov bound.}
    \label{fig:plot_FDD}
\end{figure}

\renewcommand{\arraystretch}{1.2}
\begin{table}[ht]
 \begin{center}
 \begin{tabular}{|c|c|c|}  
 \hline 
  Algorithm & $~~~e(R^*,p^s)~~~$ & $~~~R^*~~~$   \\\hline
   Lee-BJMM & 0.1618 & 0.451   \\
   Restricted Lee-BJMM for $r=5$ & 0.1539 & 0.408  \\
   Amortized Lee-BJMM & 0.1205  & 0.396 \\ 
   Amortized Restricted Lee-BJMM & 0.1189 & 0.406  \\
   Amortized Lee-Wagner & 0.1441 & 0.445  \\
   Amortized Restricted Lee-Wagner & 0.1441 & 0.445 \\ 
  \hline
  \end{tabular}
\end{center}  
  \caption{Comparison of asymptotic costs for full-distance decoding  for $p^s=47$.}\label{tab:comparison_FDD}
\end{table}

In the second scenario, where we have $N>1$ solutions, one possible technique proposed in   \cite{thomas} is to fix a rate $R \in \{0.1,\ldots, 0.9\}$ and go through all $M/2 \leq T \leq M$, to see at which $T$ the largest cost is  attained for this fixed rate. However, this approach gives for the algorithm in \cite{thomas} as well as for our algorithm always $T=M$. This is a very particular weight, where $\be$ will only have entries $\pm M$. The problem of decoding such instance is then a completely different one from the original problem and more like a binary SDP. As the algorithm in \cite{thomas} and also our algorithm work for any large $T$, they will clearly not be suitable for this special scenario.

Another possible technique is the following: the asymptotic value for $N$ is given by 
\[X=A_{(M)}(t,n,p^s)-1+R \leq R,\] thus we can fix $X$ to be a function in $R$, e.g. $X=R/2$. This will also directly lead to a $T=\lim\limits_{n \to \infty} t(n)/n$, for which $A_{(M)}(t,n,p^s)=1- R/2.$
If we would have fixed $X$ to be a constant independent of $R$ instead, this would have obstructed the comparison for all rates smaller than this constant. 
To compare the asymptotic costs of several algorithms we then determine the rate for which the cost is maximal. 
Since there is no other non-amortized algorithm which considers the second case, we will only  compare our amortized version with the algorithm provided in \cite{thomas}.

We observed that in the second case, where we decode  beyond the minimum distance, $\varepsilon$ is  very small. Note that $\varepsilon$ was introduced in \cite{rep} to increase the  number of positions on which we can merge $u$. In our algorithm, however, $u$ can be chosen very large, in fact, very close to $\ell$, even for $\varepsilon=0.$ Thus, $\varepsilon>0$ would only increase the size of the lists. We also want to note here that the program we are considering in Figure \ref{fig:plot_beyond} takes the minimum of the cost of our algorithm and the cost of brute forcing. For this note that we fixed the number of solutions to be $p^{s(k/2)}$, thus going through all vectors $\be$ of weight $t$ we expect to find a solution after $F(t,n,p^s)p^{-s(k/2)}$ many steps, that has an asymptotic cost of $A_{(M)}(t,n,p^s)-R/2= 1-R$. On the other hand, we might go through all solutions of the parity-check equations, which are $p^{sk}$ many and expect to find a solution after $p^{s(k-k/2)}$ many steps, which has an asymptotic cost of $R/2.$

\begin{figure}
    \centering
    \includegraphics[width=\textwidth]{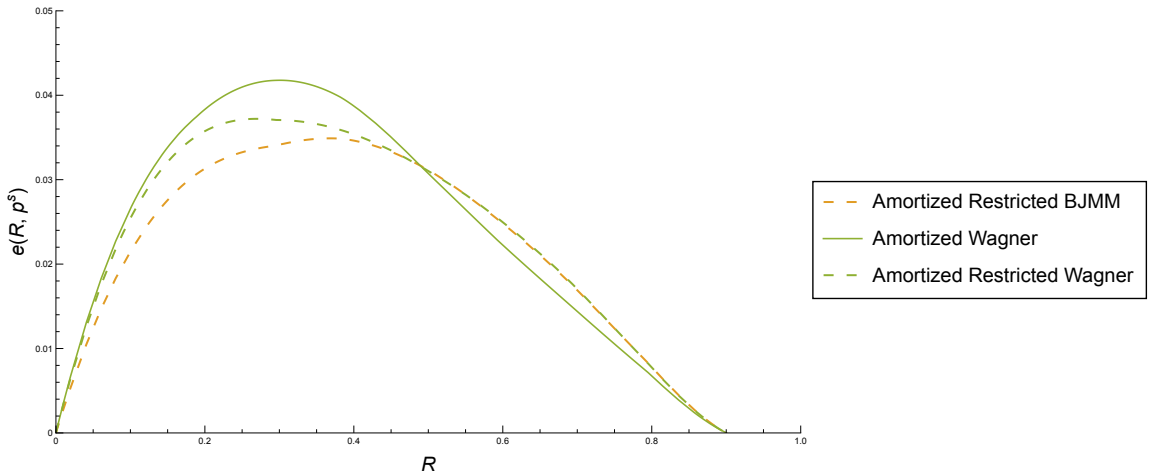}
    \caption{Comparison of asymptotic costs of decoding  beyond the minimum distance of different algorithms and their restricted versions, for $q=47$.}
    \label{fig:plot_beyond}
\end{figure}

\begin{table}[ht]
 \begin{center}
 \begin{tabular}{|c|c|c|}
 \hline 
  Algorithm & $~~~e(R^*,q)~~~$ & $~~~R^*~~~$ \\\hline
   Amortized Restricted Lee-BJMM & 0.0349 & 0.368 \\ 
      Amortized Lee-Wagner & 0.0418 & 0.301 \\
   Amortized Restricted Lee-Wagner & 0.0372 & 0.270 \\
  \hline
  \end{tabular}
\end{center}  
  \caption{Comparison of asymptotic cost of different Lee metric ISD algorithms for $p=47, s=1$  beyond the minimum distance.}\label{tab:comparison_BD}
\end{table}

\begin{remark}
This approach can work for any metric and ambient space, as long as the distribution of the error vector allows us to solve the smaller instance in a smaller space. 
This might have an impact for the RLWE problem, since also there the error vector is drawn from a certain distribution, in this case the Gaussian.
\end{remark}

\subsubsection*{Acknowledgments}
The second author is  supported by the Estonian Research Council grant number PRG49. The third author  is  supported by the Swiss National Science Foundation grant number 195290.

\bibliographystyle{plain}
\bibliography{biblio}

\end{document}